\documentclass[12pt,titlepage]{article}
\usepackage{geometry}
\usepackage{color}
\usepackage{amsmath}
\usepackage{amsthm}
\usepackage{amssymb}
\usepackage{graphicx}
\usepackage[unicode=true,
 bookmarks=false,
 breaklinks=false,pdfborder={0 0 1},backref=section,colorlinks=false]
 {hyperref}
\hypersetup{
 colorlinks,colorlinks,citecolor=blue,linkcolor=blue}

\makeatletter

\providecommand{\tabularnewline}{\\}

\graphicspath{{./figures/}} 

\usepackage{xr}
\usepackage{amsfonts}
\usepackage{array}
\usepackage[mathscr]{eucal}
\usepackage{endnotes}

  \theoremstyle{definition}
  \newtheorem{defn}{\protect\definitionname}\theoremstyle{plain}
  \newtheorem{lem}{\protect\lemmaname}\theoremstyle{plain}
  \newtheorem{cor}{\protect\corollaryname}\theoremstyle{plain}
  \newtheorem{assumption}{\protect\assumptionname}\theoremstyle{plain}
\newtheorem{thm}{\protect\theoremname}\theoremstyle{plain}
 \theoremstyle{plain}
  \newtheorem{prop}{\protect\propositionname}\theoremstyle{plain}
 \newtheorem*{lyxalgorithm*}{\protect\algorithmname}
  \theoremstyle{remark}

  \providecommand{\algorithmname}{Algorithm}
  \providecommand{\assumptionname}{Assumption}
  \providecommand{\definitionname}{Definition}
  \providecommand{\lemmaname}{Lemma}
  \providecommand{\propositionname}{Proposition}
\providecommand{\corollaryname}{Corollary}
\providecommand{\theoremname}{Theorem}

\def\1{{\mathbb 1}}

\makeatother

  \providecommand{\corollaryname}{Corollary}
  \providecommand{\definitionname}{Definition}
  \providecommand{\lemmaname}{Lemma}
  \providecommand{\propositionname}{Proposition}
  \providecommand{\remarkname}{Remark}
\providecommand{\theoremname}{Theorem}

\begin{document}

\title{Spanning Tests for \\ Markowitz Stochastic Dominance}

\author{Stelios Arvanitis\\
Athens University of Economics and Business {\normalsize{}}\thanks{{\small{}Department of Economics, 76, Patision Street, GR10434, Athens,
Greece; e-mail: stelios@aueb.gr}} \and Olivier Scaillet\\
University of Geneva and Swiss Finance Institute {\normalsize{}}\thanks{{Corresponding author: {\small{}GFRI, Bd du Pont d'Arve 40, 1211
Geneva, Switzerland; e-mail: olivier.scaillet@unige.ch}}} \and Nikolas Topaloglou\\
Athens University of Economics and Business {\normalsize{}}\thanks{{\small{}Department of International European \& Economic Studies,
76, Patision Street, GR10434, Athens, Greece; e-mail: nikolas@aueb.gr}}}

\date{First draft: January 2018\\ This draft: July 2018}
\maketitle
\begin{abstract}
We derive properties of the cdf of random variables defined as saddle-type
points of real valued continuous stochastic processes. This facilitates the 
derivation of the first-order asymptotic properties of tests for stochastic spanning given some
stochastic dominance relation. We define the concept of Markowitz
stochastic dominance spanning, and develop an analytical representation
of the spanning property. We construct a non-parametric test for spanning
based on subsampling, and derive its asymptotic exactness and consistency. 
The spanning methodology determines whether
introducing new securities or relaxing investment constraints improves
the investment opportunity set of investors driven by Markowitz stochastic
dominance. In an application to standard data sets of historical stock
market returns, we reject market portfolio Markowitz efficiency as
well as two-fund separation. Hence, we find evidence that equity
management through base assets can outperform the market, for investors
with Markowitz type preferences.
\\ \textbf{Acknowledgements}: We would like to thank the Editor, Co-Editors and the two referees for constructive criticism and numerous suggestions which have led to substantial improvements over the previous version. We thank the participants at the SFI research days 2018 for helpful comments. 
\end{abstract}

\section{Introduction}

An essential feature of any model trying to understand asset prices
or trading behavior is an assumption about investor preferences, or
about how investors evaluate portfolios. The vast majority of models
assume that investors evaluate portfolios according to the expected
utility framework. Investors are assumed to act as non -atiable and
risk averse agents, and their preferences are represented by increasing
and globally concave utility functions.

Empirical evidence suggests that investors do not always act as risk
averters. Instead, under certain circumstances, they behave in a much
more complex fashion exhibiting characteristics of both risk loving
and risk averting. They seem to evaluate wealth changes of assets
w.r.t.\ benchmark cases rather than final wealth positions. They behave
differently on gains and losses, and they are more sensitive to losses
than to gains (loss aversion). The relevant utility function can be
either concave for gains and convex for losses (S-Shaped) or convex
for gains and concave for losses (reverse S-Shaped). They seem to
transform the objective probability measures to subjective ones using
transformations that potentially increase the probabilities of negligible
(and possibly averted) events, which, in some cases, share similar analytical
characteristics to the aforementioned utility functions. Examples
of risk orderings that (partially) reflect such findings are the dominance
rules of behavioral finance (see Friedman and Savage (1948), Baucells
and Heukamp (2006), Edwards (1996), and the references therein).

Accordingly, stochastic dominance has been used over the last decades
in this framework, having more generally evolved into an important
concept in the fields of economics, finance and statistics/econometrics
(see inter alia Kroll and Levy (1980), McFadden (1989), Levy (1992),
Mosler and Scarsini (1993), and Levy (2005)), since it
enables inference on the issue of optimal choice in a non-parametric
setting. Several statistical tools have been developed to test whether, given some fixed notion of stochastic dominance,
a probability distribution of interest (or some random element that represents it) dominates any other
similar distribution in a given set, 
i.e., the former is super-efficient over the latter set (see Arvanitis et al.\ (2018)). Analogous procedures have been developed to test whether this distribution is not dominated by any other member of the given set, i.e., whether it is an efficient element of it (see Linton Post and Wang (2014)). We can find some illustrative examples  in the 
application sections of Horvath, Kokoszka, and Zitikis (2006), where
interest lies in distributions of income, or Post and Levy (2005),
Scaillet and Topaloglou (2010), Linton, Post and Whang (2014), where interest lies in financial
portfolios.

There is a large evolving literature on the first (FSD) and on the second (SSD) order stochastic dominance. 
We can characterize FSD  via the choice under uncertainty of every non-satiable investor, while we can characterize SSD by the analogous choice
of every risk averse and non-satiable investor (see Hadar and Russell
(1969), Hanoch and Levy (1969), and Rothschild and Stiglitz (1970)). Higher order stochastic dominance relations impose more restrictions on the underlying utilities of the set of investors while retaining non-satiety and risk aversion. Dropping global risk aversion, Levy and Levy (2002) formulate the notions
of prospect stochastic dominance (PSD) (see also Levy and Wiener (1998),
Levy and Levy (2004)) and Markowitz stochastic dominance (MSD). Those
notions investigate choices by investors who have S-shaped utility
functions and reverse S-shaped utility functions. Arvanitis and Topaloglou
(2017) accordingly develop consistent statistical tests for PSD and MSD super-efficiency.

Given a stochastic dominance relation, the concept of stochastic spanning
subsumes the aforementioned notion of super-efficiency. It is an idea
of Thierry Post, influenced by Mean-Variance spanning
in Huberman and Kandell (1987), that was formulated in the context
of second order stochastic dominance in Arvanitis et al.\ (2018). It is yet
generalizable to arbitrary stochastic dominance relations. Given such a relation, and if the underlying set of efficient elements, i.e., the efficient set, is non-empty, a spanning
set is simply any superset of the efficient set. As such, we can use a spanning set to provide an ''outer approximation''
of the underlying efficient set, and/or, when small enough, to provide
with a desirable reduction of the initial set of distributions upon which the stochastic dominance ordering is defined, and which
could be complicated. In such a case, we can reduce the examination
of the optimal choice problem, to a potentially easier and more parsimonious one. Both issues
are  of interest to financial economics since the underlying distributions
often represent the return behaviour of financial assets and the dominance orderings reflect classes
of investor preferences (e.g.\ for the FSD and SSD, as well as the
PSD and MSD rules and their relations to classes of utility functions,
see Levy and Levy (2002)). Those notions could also be of potential interest in any
field of economic theory or decision science that examines optimal
choice under uncertainty.

For example, if a strict subset of a universe of available assets
is known to be spanning w.r.t.\ a stochastic dominance relation that
reflects all preferences with some sort of combination of local risk
aversion with local risk seeking behavior (see for example the MSD
preorder defined in Section 3.1), any investor with such a disposition
towards risk can safely restrict her choice to the spanning set. On
the contrary, if it is not spanning, there must exist investors with
suchlike preferences that benefit from the enlargement of the investment
opportunities from the subset to the superset. This implies that stochastic
spanning can be useful in extracting important properties of
financial markets for investment decisions taylor made for particular shapes of utility functions.

Hence the following question naturally arises: for some fixed stochastic
dominance relation, is a given set of assets spanned by a (possibly
economically relevant) subset? When the two sets are not equal, spanning occurs if and only if a functional defined by
a complex recursion of optimizations is zero (see for example the
discussion in page 6 of Arvanitis et al.\ (2018) for the case of SSD,
or Proposition \ref{M-equiv} below for the case of MSD). Its empirical verification
is usually analytically intractable due to the dependence
of the functional on the generally unknown underlying distributions
and/or due to the complexity of the optimizations involved. Hence,
this is not of direct practical use. However, we can design non-parametric 
tests of the null hypothesis of spanning given the existence of empirical
information. The limit theory
of tests for stochastic spanning\footnote{Spanning tests subsume as special cases tests of super-efficiency
w.r.t.\ the underlying preorder. For example, procedures developed
in Scaillet and Topaloglou (2010), Arvanitis and Topaloglou (2017),
and Linton et al.\ (2014), can be considered as spanning tests for
singleton spanning sets.} usually involves null weak limits represented as a finite recursion
of optimization functionals applied on some relevant Gaussian process
that could have the form of a saddle-type functional. The possibility
of the existence of atoms in their distribution affects the issue
of asymptotic exactness of the aforementioned tests which are usually
based on resampling procedures such as bootstrap and subsampling (Linton
et al.\ (2005)). In order to obtain exactness, we cannot thus rely
on standard probabilistic results used in the previous work on tests of super-efficiency,
due to the complexity of the aforementioned functional.

Hence, our first contribution is the theoretical study of
continuity properties of the cdf of random variables defined as saddle type points of real valued stochastic processes. Section 2 of the
paper sets up the probabilistic framework, and derives new properties of the law of a random variable defined by
a finite number of nested optimizations on a continuous process w.r.t.\
possibly interdependent parameter spaces. Beside its usefulness for the limit theory of spanning tests developed in this paper, this result is also
a non-trivial extension to results concerning suprema of other stochastic
processes and can be useful in other econometric settings (see Section 2 for references and examples).

Our second contribution is the following. The results in Arvanitis et
al. (2018) concern the concept of stochastic spanning w.r.t.\ the SSD
relation, which essentially represents all preferences with global
risk aversion, and are derived in a context of bounded support for
the underlying distributions. We expect that analogous, yet possibly
more complex results,
on the properties of spanning sets, their representation by relevant
functionals, the construction of testing procedures, and the derivation
of their limit theory hold if we extend to local risk aversion and
general supports. Statistical tests concerning the issue of super-efficiency
w.r.t.\ stochastic dominance rules representing local attitudes towards
risk have already appeared in the literature (see for example Post
and Levy (2005), or Arvanitis and Topaloglou (2017)), but to our knowledge
the concept of spanning has not been studied yet for such dominance relations.

Section 3 investigates the concept
of stochastic spanning w.r.t.\ the MSD preorder in the context
of financial portfolios formation. We define the notion
and provide with an original characterization of spanning by the zero
of a functional. Using the principle of analogy, we define
the non-parametric test statistic, derive its limit distribution under the null hypothesis,
and define a subsampling algorithm for the approximation of the asymptotic
critical values. Among others, we use the new probabilistic
results of Section 2 and a novel combinatorial argument, for the derivation of asymptotic
exactness when the relevant limit distribution is non-degenerate and
a restriction on the significance level holds. In particular, we derive
consistency of the subsampling procedure. In contrast to the results
in Arvanitis et al.\ (2018), we allow for unbounded supports for the return
distributions, and we suppose that the relevant parameter spaces are simplicial
complexes. We explain in Section 3 why  those extensions are 
useful and how we have to modify the theoretical arguments to accommodate them.

Section 4 provides with a numerical implementation consisting of a finite
set of Linear Programming (LP) and Mixed Integer Programming (MIP) problems, the latter being
highly non linear optimization problems to solve.

Inspired by Arvanitis and Topaloglou (2017), who show that the market
portfolio is not MSD efficient, we test in an empirical application
in Section 5, whether investors with MSD preferences could beat the
market through equity management, according to Markowitz preferences.
We use equity portfolios as base assets.
We show that the market portfolio is not Markowitz efficient, and
the two-fund separation theorem does not hold for MSD investors. Thus,
combinations of the market and the riskless asset do not span the
portfolios created according to the MSD criterion. We also show that
equity managers with MSD preferences could generate portfolios that
yield 30 times higher cumulative return than the market over the last
50 years. Standard performance and risk measures  show that the optimal MSD portfolios better suit the MSD investors that 
are risk averse for losses and risk lovers for gains.
It achieves a transfer of probability mass from the left to the right 
tail of the return distribution when compared to the market portfolio.
Its return distribution exhibits less negative skewness, less kurtosis, and less negative tail risk.
Finally, using the four-factor  model of Carhart (1997) and  the  five-factor model of 
Fama and French (2015), we investigate which factors explain these returns. We find that a defensive tilt explains part 
of the performance of the optimal MSD portfolios, while momentum and profitability do not.

In the final section, we conclude. We present the proofs of the main
and the auxiliary results in the Appendix.

\section{Probabilistic Results}

Suppose that $\Lambda_{1},\Lambda_{2},\ldots,\Lambda_{s}$ are separable
metric spaces, and let $\Lambda:=\prod_{i=1}^{s}\Lambda_{i}$ be equipped
with the product topology. Consider the functional $\mathcal{\text{oper}}:=\mbox{opt}_{1}\circ\mbox{opt}_{2}\circ\cdots\circ\mbox{opt}_{s}$
where $\mbox{opt}_{i}=\sup\:\text{or}\:\inf$ w.r.t.\ to some non-empty
compact $\Lambda_{i}^{\star}\subseteq\Lambda_{i}$, for $i=1,\ldots,s$.
When $i>1$, $\Lambda_{i}^{\star}$ is allowed to depend on the elements
of $\prod_{j=1}^{i-1}\Lambda_{i-j}^{\star}$.

The probabilistic framework follows closely Chapter 2 of Nualart (2006).
It consists of a complete probability space $\left(\Omega,\mathcal{F},\mathbb{P}\right)$,
where $\mathcal{F}$ is generated by some isonormal Gaussian process
$W=\left\{ W\left(h\right),h\in H\right\} $ and $H$ is an appropriate
Hilbert space. $X$ is some vector valued stochastic process on $\Lambda$
with sample paths in the space of continuous functions $\Lambda\rightarrow\mathbb{\mathbb{R}}^{q}$
equipped with the uniform metric. In many applications, $X$ is a
Gaussian weak limit of some net of processes. We denote the Malliavin
derivative operator (see Nualart (2006)) by $D$ and by $\mathbb{D}^{1,2}$
the completion of the family of Malliavin differentiable random variables
w.r.t.\ the norm $\sqrt{\mathbb{E}\left[z^{2}+\left(Dz\right)^{2}\right]}$.

We are interested in the form of the support and the continuity properties
of the cdf of the law of the random variable $\xi:=\mathcal{\text{oper}}X_{\lambda}$.
The following assumption describes sufficient conditions for the aforementioned
law to have a countable number of atoms while being absolutely continuous
when restricted between their successive pairs. Given this, the result
to be established below, allows first for the random variable at hand
to be defined by saddle-type functionals,\footnote{The term ''saddle-type'' is used here in a somewhat abusive manner,
since commutativity between the successive optimization operations
does not hold in general.} and second for discontinuities of the resulting cdf. Hence, it generalizes
known results concerning the absolute continuity of the distribution
of suprema of stochastic processes. For an excellent treatment of
those see inter alia, Propositions 2.1.7 and 2.1.10 of Nualart (2006),
and for the discontinuities related literature on the fibering method
and its probabilistic applications, see Lifshits (1983). 
\begin{assumption}
\label{as1}For the process $X$ suppose that: 
\begin{enumerate}
\item $\mathbb{E}\left[\sup_{\Lambda}\left(X_{\lambda}^{2}\right)\right]<+\infty$. 
\item For all $\lambda\in\Lambda$, $X\left(\lambda\right)\in\mathbb{D}^{1,2}$,
and the $H$ -valued process $DX$ has a continuous version and $\mathbb{E}\left[\sup_{\Lambda}\|DX_{\lambda}\|^{2}\right]<+\infty$. 
\item For some countable $\mathcal{T}\subset\mathbb{R}$, $\mathbb{P}\left(\left\{ \xi=\tau\right\} \cap\Omega_{\tau}\right)\geq0$
holds if and only if $\tau\in\mathcal{T}$, where $\text{\ensuremath{\Omega}}_{\tau}$
denotes $\left\{ \omega\in\Omega:DX_{\lambda}\left(\omega\right)=0\mbox{ for some \ensuremath{\lambda}}\mbox{ such that }\tau=X_{\lambda}\left(\omega\right)\right\} $. 
\end{enumerate}
\end{assumption}
In the usual case where $X$ is zero-mean Gaussian,
we can establish the first condition by strong results that imply
the subexponentiality of the distribution of $\sup_{\Lambda}X_{\lambda}$,
like Proposition A.2.7 of van der Vaart and Wellner (1996). Its validity
follows from conditions that restrict the packing numbers of $\Lambda\times\mathbb{R}$
metrized as a totally bounded metric space by the use of the covariance
function of $X$, to be polynomially bounded, something that is easily
established if the $\Lambda_{i}$ are subsets of Euclidean spaces
for all $i$. In the same respect, the second condition is easily
established as in Nualart (2006) (see page 110). More specifically,
if $K\left(\lambda_{1},\lambda_{2}\right)$ is the aforementioned
covariance function, then $H$ is the closed span of $\left\{ h_{\lambda}\left(\cdot\right)=K\left(\lambda,\cdot\right),\lambda\in\Lambda\right\} $,
with inner product $\left\langle h_{\lambda_{1}},h_{\lambda_{2}}\right\rangle _{H}=K\left(\lambda_{1},\lambda_{2}\right)$,
whence $DX_{\lambda}=K\left(\lambda,\lambda\right)$. In this case,
the previous along with dominated convergence implies the existence
of $\mathbb{E}\left[\sup_{\Lambda}\|DX_{\lambda}\|^{2}\right]$. The
third condition is the most difficult to establish. In the cases that
we have in mind, we can derive ''outer approximations'' of $\mathcal{T}$ by analogous, as well as easier to establish, properties
of random variables that are stochastically dominated by $\xi$, see
for example the corollary below.

We are now able to state and prove the main probabilistic result. 
\begin{thm}
\label{aac}Under Assumption \ref{as1}, the law of $\xi$ has connected
support, say $\textnormal{supp}\left(\xi\right)$, that contains $\mathcal{T}$.
If $\tau\in\mathcal{T}$, the cdf of the law evaluated at $\tau$
has a jump discontinuity of size at most $\mathbb{P}\left(\text{\ensuremath{\Omega}}_{\tau}\right)$.
If $\tau_{1},\tau_{2}$ are successive elements of $\mathcal{T}$,
the law restricted to $\left(\tau_{1},\tau_{2}\right)$ is absolutely
continuous w.r.t.\ the Lebesgue measure. If $\mathcal{T}$ is bounded
from below then the law restricted to $\left(-\infty,\inf\mathcal{T}\right)$
is absolutely continuous w.r.t.\ the Lebesgue measure. Dually, if $\mathcal{T}$
is bounded from above then the law restricted to $\left(\sup\mathcal{T},+\infty\right)$
is absolutely continuous w.r.t.\ the Lebesgue measure. 
\end{thm}
Theorem \ref{aac} encompasses the standard absolute continuity results
in the aforementioned literature that hold when $\text{oper}$ is
a composition of suprema (or dually infima), the parameter spaces
$\Lambda$ are not dependent, and $\mathbb{P}\left(\Omega_{\tau}\right)=0$,
for all $\tau\in\mathcal{T}$. Furthermore, even in the special case
where $\mathcal{T}$ is a singleton, the result is a generalization
of Theorem 2 of Lifshits (1983) since it allows for non-Gaussianity,
dependence between the domains of the optimization operators, as well
as saddle-type optimizations. The following corollary focuses on this
particular case and estimates the size of the potential jump discontinuity
by assuming the existence of an auxiliary random variable that is
stochastically dominated by $\xi$. 
\begin{cor}
\label{cor:used}Suppose that Assumption \ref{as1} is satisfied.
Furthermore, suppose that $\mathcal{T}=\left\{ c\right\} $, $\xi\geq\eta$,
$\mathbb{P}$ a.s., and that $\text{supp}\left(\eta\right)=\left[c,+\infty\right)$.
Then, $\text{supp}\left(\xi\right)=\left[c,+\infty\right)$, its cdf
is absolutely continuous on $\left(c,+\infty\right)$, and it may
have a jump discontinuity of size at most $\mathbb{P}\left(\eta=c\right)$
at $c$. 
\end{cor}
The results above, and especially the previous corollary, are useful
for the derivation of the limit theory for our test of stochastic
spanning (see Arvanitis et al.\ (2018) for the case of SSD based on other arguments). For a given pair of sets of probability distributions driven by sets of portfolio allocations, the null hypothesis of spanning 
posits that, for any distribution in the first set, there exists some in the other one that dominates it. Below, such
a hypothesis is represented by a functional of the form
$\sup\sup\inf$ of an appropriate set of moment conditions parameterized
by such a $\Lambda$. We can obtain a test statistic through a
scaled empirical version of this functional. Under the null limit
theory for the test statistic, the results above are useful for
the construction of an asymptotically exact decision procedure based
on a resampling scheme. They do so by providing with restrictions
on the asymptotic significance level that guarantee the convergence
of the critical values to continuity points of the null limiting cdf.
In such frameworks, $X$ is usually zero-mean Gaussian, while $\xi$
is conveniently defined as a difference between infima of $X$
defined on different regions of $\Lambda$ with given properties
(see the following sections for explicit derivations of those properties in the case of
MSD).

We can meet similar probabilistic structures in other econometric
applications. An example concerns the null hypothesis of nesting of
a given statistical model by a set identified model represented by
moment inequalities. More specifically, suppose that given a random
matrix $Y$, a statistical model is comprised by a set of  probability distributions conditional
on $Y$ and parameterized by a Euclidean parameter
$\varphi\in\Phi$. A second statistical model is comprised by the
set of probability distributions conditional on $Y$  that satisfy
the conditional moment inequalities $\mathbb{E}\left[\mathbf{g}\left(\theta\right)\vert Y\right]\leq\mathbf{0}_{d}$,
for some $\theta\in\Theta$, where $\Theta$ is again a subset of
some Euclidean space, the moment function $\mathbf{g}=\left(g_{1},g_{2},\dots,g_{d}\right)$
is finite dimensional and the inequality sign is interpreted pointwisely.
We are interested in testing the hypothesis that the first model is
nested in the second model, i.e., that, for any $\varphi\in\Phi$, there
exists some $\theta\in\Theta$ such that $\mathbb{E}_{\varphi}\left[\mathbf{g}\left(\theta\right)\vert Y\right]\leq\mathbf{0}_{d}$,
where $\mathbb{E}_{\varphi}$ denotes expectation w.r.t.\ the distributions
corresponding to $\varphi$. When $\Phi$ is a singleton,
we obtain specification hypotheses similar to the ones in Guggenberger,
Hahn and Kim (2008). Under some further conditions on the properties
of $\Phi,\Theta$ and $\mathbf{g}$, the null hypothesis of nesting
is equivalent to that $\sup_{\varphi\in\Phi}\sup_{j=1,2,\dots,d}\inf_{\theta\in\Theta}\mathbb{E}_{\varphi}\left[g_{j}\left(\theta\right)/Y\right]\leq0$.
If sampling is available for any $\varphi\in\Phi$ in the first model
(this would be trivial in the specification related to the singleton case
for $\Phi$), we can form test statistics via empirical counterparts
of the functional $\sup_{\varphi\in\Phi}\sup_{j=1,2,\dots,d}\inf_{\theta\in\Theta}\mathbb{E}_{\varphi}\left[g_{j}\left(\theta\right)\vert Y\right]$. Then, the results above are also useful
for the construction of asymptotically exact decision procedures in such a context.

\section{A Spanning Test for MSD}

We now introduce the concept of stochastic spanning for the MSD relation.
We initially provide some order theoretical characterization of the
concept, and derive an analytical representation using a functional
defined by recursive optimizations. We then construct a testing procedure
using a scaled empirical counterpart of that functional and subsampling. We derive its first order limit theory 
mainly thanks to Corollary \ref{cor:used}.

\subsection{MSD and Stochastic Spanning}

Given $\left(\Omega,\mathcal{F},\mathbb{P}\right)$, suppose that
$F$ denotes the cdf of some probability measure on $\mathbb{R}^{n}$
with finite first moment.\footnote{In comparison to the spanning test for the SSD relation of Arvanitis et al.\ (2018),  we do not assume that the random variables have compact supports.} Let $G(z,\lambda,F)$ be $\int_{\mathbb{R}^{n}}\mathbb{I}\{\lambda^{Tr}$$u\leq z\}dF(u)$,
i.e., the cdf of the linear transformation $\mathbb{R}^{n}\ni x\rightarrow\lambda^{Tr}x$
where $\lambda$ assumes its values in $\mathbb{L}$ which is a closed
non-empty subset of $\mathbb{S}=\{\lambda\in\mathbb{R}_{+}^{n}:$\textbf{$\boldsymbol{1}^{Tr}\lambda$}$=1,\}$.
Analogously, let $\mathbb{K}$ denote some distinguished subcollection
of $\mathbb{L}$. In the context of financial econometrics, $F$ usually
represents the joint distribution of $n$ base asset returns, and
$\mathbb{S}$ the set of linear portfolios that can be constructed
upon the previous.\footnote{The base assets are not restricted to be individual securities but
are defined simply as the extreme points of the maximal portfolio
set $\mathbb{S}$.} The parameter set $\mathbb{L}$ represents the collection of feasible
portfolios formed by economic, legal, and/or other investment restrictions. We
denote generic elements of $\mathbb{L}$ by $\lambda,\kappa$, etc.
In order to define the concepts of MSD and subsequently of spanning, we
consider 
\[
\mathcal{J}(z_{1},z_{2},\lambda;F):=\int_{z_{1}}^{z_{2}}G\left(u,\lambda,F\right)du.
\]
\begin{defn}
\label{M-dom}$\kappa$ weakly Markowitz-dominates $\lambda$, denoted
by $\kappa\succcurlyeq_{M}\lambda$, iff 
\begin{equation}
\begin{array}{c}
\Delta_{1}\left(z,\lambda,\kappa,F\right):=\mathcal{J}\left(-\infty,z,\kappa,F\right)-\mathcal{J}\left(-\infty,z,\lambda,F\right)\leq0,\,\forall z\in\mathbb{R}_{-},\:\text{and}\\
\Delta_{2}\left(z,\lambda,\kappa,F\right):=\mathcal{J}\left(z,+\infty,\kappa,F\right)-\mathcal{J}\left(z,+\infty,\lambda,F\right)\leq0,\,\forall z\in\mathbb{R}_{++}.
\end{array}\label{eq:MSD}
\end{equation}
\end{defn}
The existence of the mean of the underlying distribution implies that
we can allow the limits of integration above to assume extended values,
hence the integral differences $\Delta_{1}$ and $\Delta_{2}$ in
(\ref{eq:MSD}) are well defined. Levy and Levy (2002) show that $\kappa\succcurlyeq_{M}\lambda$
iff the expected utility of $\kappa$ is greater than or equal to
the expected utility of $\lambda$ for any utility function in the
set of increasing and, concave on the negative part and convex on
the positive part real functions (termed as reverse S-shaped (at zero)
utility functions). Such utility functions represent preferences towards
risk that are associated with risk aversion for losses and risk loving
for gains. Hence, in financial economics,
Markowitz-dominance is the case iff portfolio $\kappa$ is weakly
prefered to portfolio $\lambda$ by every reverse S-shaped individual
investor.

The uncountable system of inequalities in (\ref{eq:MSD}) defines an order on $\mathbb{L}$. If those are satisfied as equalities,
the pair $\left(\kappa,\lambda\right)$ belongs to the possibly non-trivial
equivalence part of the order. Strict dominance $\kappa\succ_{M}\lambda$
corresponds to the irreflexive part of the order and it holds iff
at least one of the previous inequalities holds strictly for some
$z\in\mathbb{R}$, i.e., portfolio $\kappa$ is strictly prefered to
portfolio $\lambda$ by some reverse S-shaped individual investor.
Finally, given the possibility that $\Delta_{1}$ and/or $\Delta_{2}$
can change sign as functions of $z$, the relation is not generally
total. When this is the case, we cannot compare $\kappa$ and $\lambda$
w.r.t.\ $\succcurlyeq_{M}$. 

As in the Mean-Variance case, we can define the efficient set
of $\mathbb{L}$ w.r.t.\ $\succcurlyeq_{M}$, as the set of maximal
elements of the preorder. This means that $\kappa$ lies in the efficient
set iff for any $\lambda\in\mathbb{L}$, either $\kappa\succcurlyeq_{M}\lambda$
or $\kappa$ is incomparable to $\lambda$. The efficient set 
has the property that, for any $\lambda\in\mathbb{L}$, there exists
some $\kappa$ in the former for which $\kappa\succcurlyeq_{M}\lambda$.
Any superset of the efficient set has also the same
property, but the efficient set is minimal (if we ignore equivalencies)
w.r.t.\ this property. This observation motivates the definition of MSD
spanning. This is analogous to the concept
of Mean-Variance spanning introduced by Huberman and Kandel (1987),
and extended to the SSD case by Arvanitis et al.\ (2018). 
\begin{defn}
\label{M-span}$\mathbb{K}$ Markowitz-spans $\mathbb{L}$ (say $\mathbb{K}\succcurlyeq_{M}\mathbb{L}$)
iff for any $\lambda\in\mathbb{L}$, $\exists\kappa\in\mathbb{K}:\kappa\succcurlyeq_{M}\lambda$.
If $\mathbb{K=\left\{ \kappa\right\} }$, $\kappa$ is termed as Markowitz
super-efficient. 
\end{defn}
Spanning sets always exist since by construction $\mathbb{L}\succcurlyeq_{M}\mathbb{L}$.
The efficient set minimally (ignoring
equivalencies) spans $\mathbb{L}$, in the sense that any other spanning
set must be a superset of it. Hence, we can view any spanning subset of $\mathbb{L}$
 as an ``outer approximation'' of the efficient
set. Due to the complexity of (\ref{eq:MSD}) w.r.t.\ the Mean-Variance case, the mathematical properties of the efficient
set are generally difficult to derive, but fortunately, they are
approximable by properties of sequences of spanning sets that converge to it (see below).

Furthermore, if $\mathbb{K}\succcurlyeq_{M}\mathbb{L}$, the optimal
choice of every reverse S-shaped investor function lies necessarily
inside $\mathbb{K}$. Hence, if $\mathbb{K\subset\mathbb{L}}$ and
spanning occurs, we can reduce the problem of optimal choice within
$\mathbb{L}$ to the analogous problem within $\mathbb{K}$, and the
latter is more parsiminious than the former. Dually, if $\mathbb{K}$
does not span $\mathbb{L}$, there must exist optimal choices, and
thereby investment opportunities, in the increment $\mathbb{L}-\mathbb{K}$
for some MSD investors. Therefore we can motivate the interest
in the verification of spanning by tractability reasons  related to 
optimal portfolio choice, or by detection of new investment opportunities.

Super-efficiency (Arvanitis and Topaloglou
(2017)) corresponds to the existence of a greatest element for $\succcurlyeq_{M}$,
i.e., of a unique (excluding equivalencies) element that weakly Markowitz-dominates
every element of $\mathbb{L}$. Given the complexity of (\ref{eq:MSD}),
greatest elements do not generally exist. This implies that the notion of spanning not only encompasses that of super-efficiency
but it is also a property of the order that will more often hold.

The above raise the following question: given $\mathbb{\ensuremath{K}}$,
a non empty subset of $\mathbb{L}$,\footnote{We do not look at the issue of the selection of $\mathbb{\ensuremath{K}}$. Here, the latter is considered as given. In some cases, we can
select $\mathbb{\ensuremath{K}}$ by economically relevant information,
see for example the application in Arvanitis et al. (2018) for SSD.
We leave the issue of the selection of a candidate spanning
set, especially when this selection is related to the approximation
of the efficient set, for future research.} is $\mathbb{K}\succcurlyeq_{M}\mathbb{L}$? The following proposition
provides with an analytical characterization by means of nested optimizations. 
\begin{prop}
\label{M-equiv}Suppose that $\mathbb{K}$ is closed. Then $\mathbb{K}\succcurlyeq_{M}\mathbb{L}$
iff 
\begin{equation}
\xi\left(F\right):=\max_{i=1,2}\sup_{\lambda\in\mathbb{L}}\sup_{z\in A_{i}}\inf_{\kappa\in\mathbb{K}}\Delta_{i}\left(z,\lambda,\kappa,F\right)=0,\label{eq:2}
\end{equation}
where $A_{1}=\mathbb{R}_{-},\:A_{2}=\mathbb{R}_{++}$. Spanning does
not occur iff $\xi\left(F\right)>0$. 
\end{prop}
The case of super-efficiency is then trivially obtained. 
\begin{cor}
\label{superchar}Under the scope of the previous lemma, $\kappa$
is Markowitz super-efficient iff 
\[
\max_{i=1,2}\sup_{\lambda\in\mathbb{L}}\sup_{z\in A_{i}}\Delta_{i}\left(z,\lambda,\kappa,F\right)=0.
\]
\end{cor}
Given $\mathbb{\ensuremath{K}}$, it is generally difficult to directly
use the previous proposition since $F$ is usually unknown and/or
the optimizations involved are infeasible. However, given the availability
of a sample containing information for $F$ and in conjunction with
the principle of analogy, it provides the backbone for the construction
of inferential procedures that address MSD spanning.

\subsection{A Consistent Non-parametric Test}

\subsubsection{Hypotheses Structure and Test Statistic}

We employ Lemma \ref{M-equiv} to construct a non-parametric
test for MSD spanning. If $\mathbb{K}\succcurlyeq_{M}\mathbb{L}$
is chosen as the null hypothesis, the hypothesis structure takes the
form:\footnote{Corollary \ref{superchar} implies that the hypotheses are in the
special case of super-efficiency as in Arvanitis and Topaloglou (2017).} 
\[
\begin{array}{c}
\mathbf{H_{0}}:\xi\left(F\right)=0,\\
\mathbf{H_{a}}:\xi\left(F\right)>0.
\end{array}
\]

To design the decision rule,
we extend our framework as follows. \textit{\emph{Consider a}} process
$\left(Y_{t}\right)_{t\in{\mathbb{Z}}}$ taking values in $\mathbb{R}^{n}$.
$Y_{t_{i}}$ denotes the $i^{th}$ element of $Y_{t}$. The sample
of size $T$ is the random element $\left(Y_{t}\right)_{t=1,\ldots,T}$.
In our portfolio framework, it represents the observable
returns of the $n$ financial base assets. We denote
the unknown cdf of $Y_{t}$ by $F$, and the empirical cdf by $F_{T}$. We consider the test statistic
\[
\xi_{T}:=\xi\left(\sqrt{T}F_{T}\right)=\max_{i=1,2}\sup_{\lambda\in\mathbb{L}}\sup_{z\in A_{i}}\inf_{\kappa\in\mathbb{K}}\Delta_{i}\left(z,\lambda,\kappa,\sqrt{T}F_{T}\right),
\]
which is the $\sqrt{T}$-scaled empirical analog of
$\xi\left(F\right)$. We can equivalently express $\xi_{T}$ as a usual scaled empirical average: 
\begin{equation}
\xi_{T}=\max_{i=1,2}\sup_{\lambda\in\mathbb{L}}\sup_{z\in A_{i}}\inf_{\kappa\in\mathbb{K}}\frac{1}{\sqrt{T}}\sum_{t=1}^{T}q_{i}\left(z,\lambda,\kappa,Y_{t}\right),\label{eq:mp-1}
\end{equation}
where 
\[
q_{i}\left(z,\lambda,\tau,Y_{t}\right):=\begin{cases}
K\left(z,\lambda,\kappa,Y_{t}\right), & i=1\\
\left[\left(\lambda^{\prime}Y_{t}\right)_{+}-\left(\kappa^{\prime}Y_{t}\right)_{+}-v\left(z,\lambda,\kappa,Y\right)\right], & i=2
\end{cases},
\]
with $v\left(z,\lambda,\kappa,Y\right):= K\left(z,\lambda,\kappa,Y\right)-K\left(0,\lambda,\kappa,Y\right)$,
and $K\left(z,\lambda,\kappa,Y\right):= \left(z-\kappa^{\prime}Y\right)_{+}-\left(z-\lambda^{\prime}Y\right)_{+}$. This is instrumental in the numerical implementation of (\ref{eq:mp-1}) in Section 4.
When $\mathbb{K}$ is a singleton, the test statistic coincides with
the one used in Arvanitis and Topaloglou (2017). 

\subsubsection{Null Limit Distribution }

In order to show that our testing procedure is asymptotically meaningful, we need a limit theory for $\xi_{T}$ under the null
hypothesis. We derive it using the following assumption. 
\begin{assumption}
\label{MSDmix}For some $0<\delta$, $\mathbb{E}\left[\left\Vert Y_{0}\right\Vert ^{2+\delta}\right]<+\infty$.
$\left(Y_{t}\right)_{t\in{\mathbb{Z}}}$ is $a$-mixing with mixing
coefficients $a_{T}=O(T^{-a})$ for some $a>1+\frac{2}{\eta},\:0<\eta<2$,
as $T\rightarrow\infty$. Furthermore, 
\[
\mathbb{\mathbb{\mathbb{V}=E}}\left[\left(Y_{0}-\mathbb{\mathbb{E}}Y_{0}\right)\left(Y_{0}-\mathbb{\mathbb{E}}Y_{0}\right)^{T}\right]+2\sum_{t=1}^{\infty}\mathbb{\mathbb{E}}\left[\left(Y_{0}-\mathbb{\mathbb{E}}Y_{0}\right)\left(Y_{t}-\mathbb{\mathbb{E}}Y_{t}\right)^{T}\right]
\]
is positive definite. 
\end{assumption}

The mixing rates condition is implied by stationarity
and geometric ergodicity. The latter holds for many stationary models
used in the context of financial econometrics, like 
ARMA, GARCH-type, and stochastic volatility models (see Francq
and Zakoian (2011) for several examples). The moment existence condition
enables the validity of a mixing CLT. A CLT typically holds under stricter restrictions. The positive definiteness of the long run covariance
matrix is for instance satisfied, if $\left(Y_{t}\right)_{t\in{\mathbb{Z}}}$
is a vector martingale difference process and the elements of $Y_{0}$
are linearly independent random variables. From the compactness of
$\mathbb{L}$, the previous implies that $\sup_{\lambda\in\mathbb{L}}\int_{-\infty}^{+\infty}\sqrt{G\left(u,\lambda,F\right)\left(1-G\left(u,\lambda,F\right)\right)}du<+\infty,$
which is a uniform version of the analogous condition used in Horvath et al.\ (2006).

We establish the limit theory below via the use of the concept of
Skorokhod representations along with an iterative consideration of
the dual notions of epi/hypo-convergence. The result depends on the
contact sets 
\[
\Gamma_{i}=\left\{ \lambda\in\mathbb{L},\kappa\in\mathbb{K},z\in A_{i}:\Delta_{i}\left(z,\lambda,\kappa,F\right)=0\right\} .
\]
For any $i$, $\Gamma_{i}$ is non empty since $\Gamma_{i}^{\star}\equiv\left\{ \left(\kappa,\kappa,z\right),\kappa\in\mathbb{K},z\in A_{i}\right\} \subseteq\Gamma_{i}$.
Furthermore, if the support of $F$ is bounded, for any $\lambda\in\mathbb{L},\kappa\in\mathbb{K},\:\exists z\in A_{i}:\left(\lambda,z\right)\in\Gamma_{i}$,
for all $i=1,2$,.\footnote{For example, since the support is bounded, we can cover it by some
hypercube of the form $\left[z_{l},z_{u}\right]^{n}$ where we can
choose $z_{l}$ as negative. Obviously, $\left(\lambda,z_{l}\right)\in\Gamma_{1},$
for any $\lambda\in\mathbb{L}$. } Hence, $\Gamma_{i}^{\star}\subset\Gamma_{i}$.

In what follows, we denote convergence in distribution by $\rightsquigarrow$. 
\begin{prop}
\label{EAD}Suppose that $\mathbb{K}$ is closed, Assumption \ref{MSDmix}
holds, and \textup{$\mathbf{H_{0}}$} is true. Then as $T\rightarrow\infty$,
$\xi_{T}\rightsquigarrow\xi_{\infty},$ where $\xi_{\infty}:=\max_{i=1,2}\sup_{z\in A_{i}}\sup_{\lambda}\inf_{\kappa}\Delta_{i}\left(z,\lambda,\kappa,\mathcal{G}_{F}\right),\:\left(\lambda,z,\kappa\right)\in\Gamma_{i},$
and $\mathcal{G}_{F}$ is a centered Gaussian process with covariance
kernel given by \linebreak{}
 $\text{Cov}(\mathcal{G}_{F}(x),\mathcal{G}_{F}(y))=\sum_{t\in\mathbb{Z}}\text{Cov}\left(\mathbb{I}_{Y_{0}\leq x},\mathbb{I}_{Y_{t}\leq y}\right)$
and $\mathbb{P}$ almost surely uniformly continuous sample paths
defined on $\mathbb{R}^{n}$.\footnote{See Theorem 7.3 of Rio (2013).} 
\end{prop}
The covariance kernel above, and thereby $\mathcal{G}_{F}$, are well
defined due to the mixing condition and the existence of $\sup_{\lambda\in\mathbb{L}}\int_{-\infty}^{+\infty}\sqrt{G\left(u,\lambda,F\right)\left(1-G\left(u,\lambda,F\right)\right)}du$
implied by Assumption \ref{MSDmix} (see Remark 1 in Arvanitis and
Topaloglou (2017)).

\subsubsection{A Subsampling Based Testing Procedure: Limit Theory and Combinatorial
Considerations}

We cannot directly use the result in Proposition \ref{EAD} for the
construction of an asymptotic decision rule since the distribution
of $\xi_{\infty}$ depends on the unknown covariance kernel of $\mathcal{G}_{F}$.
We can establish a feasible decision rule by the use of a resampling
procedure. We consider subsampling, as in Linton et al.\ (2014)-see also Linton et al.\
(2005). This resampling is of a non-parametric nature since we do not want to specify parametric conditional distributions for the multivariate return dynamics. 

\begin{lyxalgorithm*}
\label{Sub_alg}The testing procedure consists of the following steps: 
\begin{enumerate}
\item \emph{Evaluate $\xi_{T}$ at the original sample value. } 
\item \emph{For $0<b_{T}\leq T$ , generate subsample values from the original
observations $(Y_{i})_{i=t,\ldots t+b_{T}-1}$ for all $t=1,2,\ldots,T-b_{T}+1$.} 
\item \emph{Evaluate the test statistic on each subsample value, obtaining
$\xi_{T,b_{T},t}$ for all $t=1,2,\ldots,T-b_{T}+1$. } 
\item \emph{Approximate the cdf of the asymptotic distribution of $\xi_{T}$
by}\\
\emph{$s_{T,b}(y)=\frac{1}{T-b_{T}+1}\sum_{t=1}^{T-b_{T}+1}1\left(\xi_{T,b_{T},t}\leq y\right)$
and evaluate its $1-\alpha$ quantile $q_{T,b_{T}}\left(1-\alpha\right)$.} 
\item \emph{$\text{{Reject}\:}\mathbf{H_{0}}\:\text{{iff}\:}\xi_{T}>q_{T,b_{T}}(1-\alpha).$}
\end{enumerate}
\end{lyxalgorithm*}
We derive the first order limit theory via the use of Proposition
\ref{EAD} and of relevant results from the theory of subsampling.
We first make the following standard assumption in the subsampling
methodology. 
\begin{assumption}
\label{subseq}Suppose that $\left(b_{T}\right)$, possibly depending
on $\left(Y_{t}\right)_{t=1,\ldots,T}$, satisfies 
\[
\mathbb{P}\left(l_{T}\leq b_{T}\leq u_{T}\right)\rightarrow1,
\]
where $(l_{T})$ and $(u_{T})$ are real sequences such that $1\leq l_{T}\leq u_{T}$
for all $T$, $l_{T}\rightarrow\infty$ and $\frac{u_{T}}{T}\rightarrow0$
as $T\rightarrow\infty$. 
\end{assumption}
The assumption does not provide with much information on the practical
choice of the subsampling rate for fixed $T$. It is designed to handle
issues like asymptotic exactness and consistency. In the following
section, along with the numerical implementation for\emph{ $\xi_{T}$},
we discuss a method of fixed $T$ correction for the algorithm above,
in the spirit of Arvanitis et al. (2018), that involves the use of
several subsampling rates.

Asymptotic exactness is derivable by results like Theorem 3.5.1 in
Politis, Romano and Wolf (1999). The latter requires continuity of
the limit cdf at the quantile corresponding to the significance level
$\alpha$. Even when the distribution of $\xi_{\infty}$ is non-degenerate,
it is possible that it has a cdf with a unique discontinuity at zero
(see the proof of Lemma \ref{aac-1} in the Appendix). If there exists
a lower bound for $\xi_{\infty}$, Corollary \ref{cor:used} provides
with an estimate for the cdf jump size at zero. Then the use of
the aforementioned theorem  becomes possible by properly restricting
$\alpha$. This is where the new probabilistic results of Section 2
become useful in our context. It turns out (see the proof of Lemma \ref{aac-1} in
the Appendix) that we can obtain such a bound in the form of a non-negative
random variable defined as the difference between the suprema at $\mathbb{L}$
and $\mathbb{K}$ respectively, of a linear Gaussian process. Hence,
we get the needed estimate of the jump size as the probability that
the latter random variable attains the value zero. 

In order to evaluate this, we essentially use some combinatorial notions
that allow the estimation of the proportion of the linear functions
for which their unique maximizer over $\mathbb{L}$ is
a common extreme point of both the parameter spaces. 
\begin{defn}
\label{def:eff_extreme_points}Suppose that $\mathbb{M},\mathbb{N}$
are simplicial complexes inside $\mathbb{S}$ and $\mathbb{M}\supseteq\mathbb{N}$.
The set of effective extreme points of $\mathbb{N}$ w.r.t.\ $\mathbb{M}$
is 
\[
e_{\mathbb{M}}\left(\mathbb{N}\right):=\left\{ \lambda\:\text{is an extreme point of }\mathbb{N}:\exists\text{ extreme point \ensuremath{s} of \ensuremath{\mathbb{S}}:}\left\Vert \lambda-s\right\Vert \leq\inf_{\kappa\in\mathbb{M}}\left\Vert \kappa-s\right\Vert \right\} .
\]
Furthermore, if $\lambda\in e_{\mathbb{M}}\left(\mathbb{N}\right)$
then the set of the adjoint to $\lambda$ extreme points of $\mathbb{S}$
is 
\[
c\left(\lambda\right):=\left\{ s\:\text{is an extreme point of }\mathbb{S}:\left\Vert \lambda-s\right\Vert \leq\inf_{\kappa\in\mathbb{M}}\left\Vert \kappa-s\right\Vert \right\} .
\]
\end{defn}
Given the non-linear simplicial complex forms of $\mathbb{M},\mathbb{N}$,
the notion of an effective extreme point essentially picks the extreme
points of $\mathbb{N}$ that can be restricted to $\mathbb{M}$ maximizers
of linear real functions defined on $\mathbb{S}$. Given any such
extreme point, its adjoint set essentially picks up the extreme points
of the incorporating simplex $\ensuremath{\mathbb{S}}$ that are closer
to it than any other extreme point of $\mathbb{M}$.
\begin{defn}
\label{def:char}The $\mathbb{M}$-character of $\lambda\in e_{\mathbb{M}}\left(\mathbb{N}\right)$
w.r.t.\ $s\in c\left(\lambda\right)$ is 
\[
ch_{\mathbb{M}}\left(s,\lambda\right):=\#\left\{ \kappa\in e_{\mathbb{M}}\left(\mathbb{N}\right):\left\Vert \lambda-s\right\Vert =\left\Vert \kappa-s\right\Vert \right\} .
\]
Furthermore, the $\mathbb{M}$-character of $\mathbb{N}$ is 
\begin{equation} \label{Mchara}
ch_{\mathbb{M}}\left(\mathbb{N}\right):=\sum_{\lambda\in e_{\mathbb{M}}\left(\mathbb{N}\right)}\sum_{s\in c\left(\lambda\right)}\frac{\left(n-ch_{\mathbb{M}}\left(s,\lambda\right)\right)!}{n!}.
\end{equation}
\end{defn}
The ratio $\frac{\left(n-ch_{\mathbb{M}}\left(s,\lambda\right)\right)!}{n!}$
counts the proportion of linear real functions with unique maximizer $s$ 
over $\mathbb{S}$, and unique maximizer
$\lambda$ over the restricted $\mathbb{M}$. Hence, $ch_{\mathbb{M}}\left(\mathbb{N}\right)$
counts the proportion of such functions for which the maximizer over $\mathbb{M}$  is an extreme point of $\mathbb{N}.$
Suppose now that $Z$ follows a non-degenerate, zero mean, $n$-dimensional
Normal distribution. The characterization  (\ref{Mchara}) of the $\mathbb{M}$-character
of $\mathbb{N}$ allows the bounding from above (see the proof of
Lemma \ref{aac-1} in the Appendix) of the
probability of the event $\sup_{\lambda\in\mathbb{M}}\lambda'Z=\sup_{\lambda\in\mathbb{N}}\lambda'Z$
by $ch_{\mathbb{M}}\left(\mathbb{N}\right)$, and this is directly
related to the estimation of the potential jump size discontinuity
of the cdf of $\xi_{\infty}$ at zero. Thereby, if we assume that $\mathbb{L}$
and $\mathbb{K}$ are simplicial complexes, and if $ch_{\mathbb{L}}\left(\mathbb{K}\right)$
is easy to evaluate, the previous definitions greatly facilitate and are key for
the derivations of the asymptotic exactness of our test. 
\begin{assumption}
\label{simpl_span}$\mathbb{L}$ and $\mathbb{K}$ are simplicial
complexes inside the standard simplex $\mathbb{S}=\{\lambda\in\mathbb{R}_{+}^{n}:1^{\prime}$\textbf{$\lambda$}$=1,\}$
and $e_{\mathbb{L}}\left(\mathbb{K}\right)\subset e_{\mathbb{L}}\left(\mathbb{L}\right)$. 
\end{assumption}
The simplicial form of $\mathbb{L}$ and $\mathbb{K}$
generalizes considerably the setting of Arvanitis et al.\
(2018). There, those spaces are restricted as convex polytopes. Here,
they need not be convex, and they can be disconnected, discrete, etc.
 This could be useful when the investment categories are constrained because of SRI screening, restrictions on foreign investment, restrictions on available type of shares, etc. This generalization allows for the establishment of the asymptotic
validity and thereby the applicability of our test in more complicated
scenarios. For example, suppose that $\mathbb{K}=\mathbb{K}_{1}\cup\mathbb{K}_{2}$
which are disjoint simplices. If $\mathbb{K}\succcurlyeq_{M}\mathbb{L}$,
but neither $\mathbb{K}_{1}\succcurlyeq_{M}\mathbb{L}$, nor $\mathbb{K}_{2}\succcurlyeq_{M}\mathbb{L}$,
then this directly implies that the efficient set is disconnected.
If $e_{\mathbb{L}}\left(\mathbb{K}\right)\subset e_{\mathbb{L}}\left(\mathbb{L}\right)$
holds, Assumption \ref{simpl_span} holds for the pairs $\left(\mathbb{L},\mathbb{K}\right)$,
$\left(\mathbb{L},\mathbb{K}_{1}\right)$, $\left(\mathbb{L},\mathbb{K}_{2}\right)$.
Then, we can use the test to determine the spanning relations
inside each pair and thereby determine the disconnectedness
of the efficient set. If the convex polytope case is not
generalized as in this paper, the determination of the spanning relation
between the elements of the first pair is not feasible.

Assumption \ref{simpl_span} implies that $e_{\mathbb{L}}\left(\mathbb{L}\right)$
is finite. The $e_{\mathbb{L}}\left(\mathbb{K}\right)\subset e_{\mathbb{L}}\left(\mathbb{L}\right)$
part implies $ch_{\mathbb{L}}\left(\mathbb{K}\right)\leq1$. Indicative
examples are the following. First, consider the trivial case where
$\mathbb{K}$ is interior to $\mathbb{L}.$ Then, it is obvious that
$ch_{\mathbb{L}}\left(\mathbb{K}\right)=0$. Second, consider the
case where $\mathbb{L=S}$, and $e_{\mathbb{L}}\left(\mathbb{K}\right)\neq\emptyset$
and Assumption \ref{simpl_span} holds. Then, $ch_{\mathbb{L}}\left(\mathbb{K}\right)=\frac{\#e{}_{\mathbb{L}}\left(\mathbb{K}\right)}{n}<1.$
Finally, suppose that $n=3$, $\mathbb{L}$ is a line in the interior
of the triangle, such that each boundary point of the line has a minimal
distance from a unique triangle vertex and that both boundary points
have the same distance from the remaining vertex. Furthermore, suppose
that $\mathbb{K}$ is some half of that line. Then, $e_{\mathbb{L}}\left(\mathbb{L}\right)$
consists of both the line boundary points, and $e_{\mathbb{L}}\left(\mathbb{K}\right)$
consists of the boundary point that lies in the chosen half. If $\lambda$
is an effective extreme point in either set, the cardinality of $c\left(\lambda\right)$
equals two. Moreover, $ch_{\mathbb{L}}\left(s,\lambda\right)$ equals
1 if $s$ lies closer to $\lambda$ than to the other boundary point
of the line, and equals 2 in the other case. Hence, $ch_{\mathbb{L}}\left(\mathbb{K}\right)=\frac{1}{2}$.

Given our assumptions and the new combinatorial arguments not used previously in the literature,  
we prove in the Appendix (see Lemma \ref{aac-1}) that the probability
that the aforementioned bounding random variable attains the value
zero is less than or equal to $ch_{\mathbb{L}}\left(\mathbb{K}\right)$. Then, via the use of Corollary \ref{cor:used}, we establish
that, when $\xi_{\infty}$ is non degenerate, the $1-\alpha$ quantile
is a continuity point for its cdf when $\alpha<1-ch_{\mathbb{L}}\left(\mathbb{K}\right)$.
Hence, we immediately obtain the following first order limit theory
for the subsampling testing procedure described above via Theorem
3.5.1 in Politis, Romano and Wolf (1999). 
\begin{thm}
\label{main2} Suppose that Assumptions \ref{MSDmix}, \ref{subseq}
and \ref{simpl_span} hold. For the testing procedure described in
Algorithm \ref{Sub_alg}, we have that 
\end{thm}
\begin{enumerate}
\item \emph{If $\mathbf{H_{0}}$ is true and $\xi_{\infty}$ is constant,
then, 
\[
\lim_{T\rightarrow\infty}\mathbb{P}\left(\xi_{T}>q_{T,b_{T}}\left(1-\alpha\right)\right)=0.
\]
}
\item \emph{If $\mathbf{H_{0}}$ is true, $\xi_{\infty}$ is non-constant,
and $\alpha<1-ch_{\mathbb{L}}\left(\mathbb{K}\right)$, then}, 
\[
\lim_{T\rightarrow\infty}\mathbb{P}\left(\xi_{T}>q_{T,b_{T}}\left(1-\alpha\right)\right)=\alpha.
\]
\item \emph{If $\mathbf{H_{a}}$ is true, then,} 
\[
\lim_{T\rightarrow\infty}\mathbb{P}\left(\xi_{T}>q_{T,b_{T}}\left(1-\alpha\right)\right)=1.
\]
\end{enumerate}
When the distribution of $\xi_{\infty}$ is degenerate, the procedure
is asymptotically conservative even if the restriction $\alpha<1-ch_{\mathbb{L}}\left(\mathbb{K}\right)$
does not hold. This is reminiscent of the results in Linton et al.\ (2005) concerning testing procedures for superefficiency
w.r.t.\ several stochastic dominance relations. The non-degeneracy
of the aforementioned limit distribution is not easy to establish
except for cases such as the one about bounded supports which was
discussed above.

When the distribution of $\xi_{\infty}$ is non-degenerate, the procedure
is asymptotically exact if the restriction $\alpha<1-ch_{\mathbb{L}}\left(\mathbb{K}\right)$
holds. The restriction on the significance level is non-binding in
usual applications. For example, when $\mathbb{L}=\mathbb{S}$ and
$\mathbb{K}$ is a singleton, i.e., when the test is applied for super-efficiency,
it implies at worst that $\alpha<1/2$, something that is usually
satisfied. The closer to binding the restriction becomes, the more
extreme points of $\mathbb{L}=\mathbb{S}$ exist inside $\mathbb{K}$.
An extreme case is when $n$ is large, $\mathbb{K}$ is finite, and
contains $n-1$ extreme points. In such a case, the result leads to
subsampling tests that tend to asymptotically favor the  null
hypothesis of spanning. We could handle that by breaking up $\mathbb{K}$ is ''smaller
pieces'' and iterating the testing procedure w.r.t.\ them. For example,
we can apply the procedure for any subset of $\mathbb{K}$ that contains
$m$ points, for $m$ sufficiently small in order to obtain a meaningful
significance level. If for some subset, we cannot reject spanning,
we can infer that we cannot reject spanning for the initial $\mathbb{K}$,
since supersets of spanning sets are spanning sets from Definition
\ref{M-span}. It is also possible that the structure of the efficient
set prohibits such a $\mathbb{K}$ to be a spanning set. We leave
the study of such questions for future work. In any case, the testing
procedure is consistent.

Under some assumptions, we can prove, using again among others the
main result, that an analogous testing procedure based on block bootstrap
is generally asymptotically conservative and consistent.

\section{A Numerical Implementation and Bias Correction}

We first describe a potential numerical implementation via the use
of a testing procedure asymptotically equivalent to the one of Subsection
\ref{Sub_alg}, and obtained by finite approximations of the $A_{i},\:i=1,2$,
as well as applications of mixed integer and linear programming. For
each $T$, let $A_{i}^{\left(T\right)}$ denote a finite subset of
$A_{i}$ for each $i$. Then consider the test statistic defined by
\[
\xi_{T}^{\star}:=\max_{i=1,2}\sup_{\lambda\in\mathbb{L}}\sup_{z\in A_{i}^{\left(T\right)}}\inf_{\kappa\in\mathbb{K}}\Delta_{i}\left(z,\lambda,\kappa,\sqrt{T}F_{T}\right),
\]
and modify the algorithm of Subsection \ref{Sub_alg} by using $\xi_{T}^{\star}$
in place of $\xi_{T}$. Under the previous assumption framework if,
as $T\rightarrow+\infty$, $A_{i}^{\left(T\right)}$ appropriately
approximates $A_{i}$, the modified procedure has the same first order
limit theory with the original one. 
\begin{thm}
\label{thm:mod}Suppose that Assumptions \ref{MSDmix}, \ref{subseq}
and \ref{simpl_span} hold. If, as $T\rightarrow+\infty$, $A_{i}^{\left(T\right)}$
converges to some dense subset of $A_{i}$ in Painleve-Kuratowski
sense for all $i=1,2$, the results of Theorem \ref{main2} hold also
for the modified procedure. 
\end{thm}
Now, the integration by parts formula for Lebesgue-Stieljes integrals
and the commutativity of suprema imply that 
\begin{equation}
\xi_{T}^{\star}=\max_{i=1,2}\sup_{z\in A_{i}^{\left(T\right)}}\sup_{\lambda\in\mathbb{L}}\inf_{\kappa\in\mathbb{K}}\frac{1}{\sqrt{T}}\sum_{t=1}^{T}q_{i}\left(z,\lambda,\kappa,Y_{t}\right),\label{eq:mp-2}
\end{equation}
where the $q_{i}$ are defined in \ref{eq:mp-1}. From the finiteness
of $A_{i}^{\left(T\right)},\:i=1,2$, the non trivial parts of the
optimizations involved concern the $n_{i,T}:=\sup_{\lambda\in\mathbb{L}}\inf_{\kappa\in\mathbb{K}}\frac{1}{\sqrt{T}}\sum_{t=1}^{T}q_{i}\left(z,\lambda,\kappa,Y_{t}\right)$.
Furthermore, 
\[
n_{1,T}=\inf_{\kappa\in\mathbb{K}}\frac{1}{\sqrt{T}}\sum_{t=1}^{T}\left(z-\kappa^{\prime}Y\right)_{+}-\inf_{\lambda\in\mathbb{L}}\frac{1}{\sqrt{T}}\sum_{t=1}^{T}\left(z-\lambda^{\prime}Y\right)_{+},
\]
and we can reduce each of the minimizations involved to the solution
of linear programming problems.

There is a set of at most $T$ values, say ${\cal R}=\{r_{1},r_{2},...,r_{T}\}$,
containing the optimal value of the variable $z$ (see Scaillet and
Topaloglou (2010) for the proof). Thus, we solve smaller problems
$P(r)$, $r\in{\cal R}$, in which $z$ is fixed to $r$. Now, each
of the above minimization problems boils down to a linear problem.
Without loss of generality, the first optimization problem is the
following:

\begin{subequations} \label{model-SSD2} 
\begin{eqnarray}
\mathtt{\min} &  & \sum_{t=1}^{T}W_{t}\nonumber \\
\mathtt{\mbox{s.t.}} &  & W_{t}\geq r-\kappa^{\prime}Y_{t},\quad\forall t\in T\nonumber \\
 &  & e^{\prime}\kappa=1,\nonumber \\
 &  & \kappa\geq0,\nonumber \\
 &  & W_{t}\geq0,\quad\forall t\in T.
\end{eqnarray}
\end{subequations}

Furthermore, and via the results in the first Appendix of Arvanitis
and Topaloglou (2017), we have that 
\[
n_{2,T}=\sup_{\lambda\in\mathbb{L}}\frac{1}{\sqrt{T}}\sum_{t=1}^{T}\max\left(\lambda'Y_{t},z\right)-\sup_{\kappa\in\mathbb{K}}\frac{1}{\sqrt{T}}\sum_{t=1}^{T}\max\left(\kappa'Y_{t},z\right).
\]
Hence, we need to solve both optimization problems appearing above.
We do so via representing them as MIP programs. Again, there is a
set of $T$ values, say ${\cal R^{\prime}}=\{r_{1}^{\prime},r_{2}^{\prime},...,r_{T}^{\prime}\}$,
containing the optimal value of the variable $z$ (see Arvanitis and
Topaloglou (2017) for the proof). Thus, we solve smaller problems
$P(r)$, $r\in{\cal R^{\prime}}$, in which $z$ is fixed to $r$.
Consider without loss of generality the first optimization problem:

\label{model-PSDP} 
\begin{eqnarray}
\mathtt{\max_{\lambda\in\mathbb{L}}} &  & \frac{1}{\sqrt{T}}\sum_{t=1}^{T}(X_{t}-cb_{t})\nonumber \\
\mathtt{\mbox{s.t.}} &  & X_{t}=\lambda'Y_{t}b_{t}+r(1-b_{t})\quad\forall\,t\in T,\label{MSDP1}\\
 &  & r-\lambda'Y_{t}+Mb_{t}>0\quad\forall\,t\in T,\label{MSDP2}\\
 &  & \lambda^{\prime}\boldsymbol{1}=1,\label{MSDP3}\\
 &  & \lambda\geq0,\label{MSDP4}\\
 &  & b_{t}\in\{0,1\}\quad\forall\,t\in T.\label{MSDP5}
\end{eqnarray}

Hence, the computational cost of the implementation above consists
of $\text{card}A_{1}$ linear programming problems, $\text{card}A_{2}$
mixed integer programming problems, and three trivial optimizations.

Secondly, and although the tests above have asymptotically correct
size, it is expected that the quantile estimates $q_{T,b_{T}}(1-\alpha)$
may be biased and sensitive to the subsample size $b_{T}$ in finite
samples of realistic dimensions for $n$ and $T$. To correct for
small-sample bias and reduce the sensitivity to the choice of $b_{T}$,
we follow Arvanitis et al. (2018). For a given significance level
$\alpha$, we compute the quantiles $q_{T,b_{T}}(1-\alpha)$ for a
range of values for the subsample size $b_{T}$. We then estimate
the intercept and slope of the following regression line using OLS
regression analysis:

\begin{equation}
q_{T,b_{T}}(1-\alpha)=\gamma_{0;T,1-\alpha}+\gamma_{1;T,1-\alpha}(b_{T})^{-1}+\nu_{T;1-\alpha,b_{T}}.\label{eq:OLS regression line}
\end{equation}
We then estimate the bias-corrected $(1-\alpha)$-quantile as the
OLS predicted value for $b_{T}=T$:

\begin{equation}
q_{T}^{BC}(1-\alpha):=\hat{\gamma}_{0;T,1-\alpha}+\hat{\gamma}_{1;T,1-\alpha}(T)^{-1}.\label{eq:estimated critical value}
\end{equation}
Since $q_{T,b_{T}}(1-\alpha)$ converges in probability to $q(\xi_{\infty},1-\alpha)$
and $(b_{T})^{-1}$ converges to zero as $T\rightarrow0$, $\hat{\gamma}_{0;T,1-\alpha}$
converges in probability to $q(\xi_{\infty},1-\alpha)$, and the asymptotic
properties are not affected.

\section{Monte Carlo Study}

We now design Monte Carlo experiments to evaluate
the size and power of our testing procedure in finite samples. We allow for
conditional heteroskedasticity consistent with empirical findings on returns of financial data as observed in the empirical application below. The multivariate return process $\left(Y_{t}\right)_{t\in\mathbb{Z}}$
is a vector GARCH(1,1) process, which is transformed to accommodate both spanning (size) and non spanning cases (power) for $K$ given assets. Such a process permits both temporal and cross sectional dependence between the random variables stacked in the vector process.

Suppose that $(z_{t}), t\in\mathbb{Z}$, are i.i.d. with mean zero, unit variance, and $\mathbb{E}\left[\vert z_{t}\vert^{2+\epsilon}\right]<\infty$, for some $\epsilon>0$. 
We assume that the cdf of $z_{t}$ is strictly increasing. Furthermore,
  we define the components of the return process for $i=1,...,K-1$ as  
\begin{eqnarray*}
y_{i,t} & = & \mu_{i}+z_{t}h_{i,t}^{1/2}\text{,}\\
h_{i,t} & = & \omega_{i}+\left(a_{i}z_{t-1}^{2}+\beta_{i}\right)h_{i,t-1}\text{, }
\end{eqnarray*}
with $\mathbb{E}\left[a_{i}z_{t}^{2}+\beta_{i}\right]^{1+\epsilon}<1$, for some $\epsilon>0$, and  $\omega_{i},a_{i},\beta_{i}\in\mathbb{R}{}_{++}$,
$\mu_{i}\in\mathbb{R}{}_{+}$. For asset $i=K$, we define \[
y_{K,t}=v_{1}\left(z_{t}h_{K-1,t}^{1/2}\right)_{+}+v_{2}\left(z_{t}h_{K-1,t}^{1/2}\right)_{-},
\]
with $v_{1},v_{2}\in\mathbb{R}$.

Let $\mathbf{\tau=}\left(0,0,...,1,0\right)$,
$\mathbf{\tau}^{\star}\mathbf{=}\left(0,0,0,...,1\right)$, and $\mathbb{L}:=\left\{ \left(\lambda,0,0,\right)^{Tr},\mathbf{\tau,\tau}^{\star}\right\}$, with $\lambda\in\mathbb{R}{}_{+}^{K-2}$ and $1^{Tr}\lambda=1$.
Using this portfolio space, we obtain the following result on Markowitz-spanning.
\begin{prop}
\label{optimal1}If $\mu_{i}=0$ for $i=1,...,K-1$, $\left\vert v_{1}\right\vert >\sqrt{\frac{\max\left\{ \omega_{i},a_{i},\beta_{i}\text{, }i=1,...,K-1\right\} }{\min\left\{ \omega_{i},a_{i},\beta_{i}\text{, }i=1,...,K-1\right\} }}$
and $\left\vert v_{2}\right\vert <\sqrt{\frac{\min\left\{ \omega_{i},a_{i},\beta_{i}\text{, }i=1,...,K-1\right\} }{\max\left\{ \omega_{i},a_{i},\beta_{i}\text{, }i=1,...,K-1\right\} }}$,
then for $M \leq K-2$, the subset $\mathbb{\mathbb{K}}:=\left\{ \left(\lambda,0,0\right)^{Tr},\tau^{\star}\right\}$
with $\lambda\in\mathbb{R}{}_{+}^{M}$ and $1^{Tr}\lambda=1$, Markowitz-spans $\mathbb{L}$, while $\mathbb{\mathbb{K}}\setminus \left\{ \tau^{\star}\right\} $
does not Markowitz-span $\mathbb{L}$. 
\end{prop}

The statement of Proposition \ref{optimal1} extends  Proposition 4 of Arvanitis and Topaloglou
(2017) to allow for $K$ assets with any subset of $M$ spanning assets, as well as non Gaussian innovations. Its proof follows the same arguments as in  Arvanitis and Topaloglou
(2017), and is thus omitted. 
It depends on $\mathbf{\tau}^{\star}$
being a Markowitz super-efficient portfolio w.r.t.\ the portfolio space.  The design of Monte Carlo experiments in a dynamic setting 
is not easy for our testing procedure since we need to work with stationary distributions and different assets.
The properties of those distributions required to show spanning and no spanning results are often difficult to characterize.\footnote{Another example is a process with different positive means and no serial dependence  such that
\begin{eqnarray*}
y_{i,t} & = & \mu_i + z_t, \qquad i =1,..,K-1,  \\
y_{K,t} & = &   v_1 \mu_{K-1} 1_{z_t > 0} + v_2 \mu_{K-1} 1_{z_t < 0}  + z_t ,
\end{eqnarray*}
with $\mu_i>0$,  $i =1,...,K-1$.
Then, if $\displaystyle v_2  > \frac{\max\{\mu_i, i= 1,...,K-1\}}{\min\{\mu_i, i= 1,...,K-1\}}$ and 
 if $\displaystyle 0 < v_1  <\frac{\min\{\mu_i, i= 1,...,K-1\}}{\max\{\mu_i, i= 1,...,K-1\}},$ the spanning results stated in Proposition \ref{optimal1} also hold. We have checked in unreported simulation results that the spanning test behaves also well in such an example including the case of student innovations with infinite variance. 
}

We present our Monte Carlo results in Table \ref{tbl:sampling}. The number of replications to compute the empirical size and power is 1000 runs.
We use either a combination of 2 assets ($M=2$) plus portfolios $\mathbf{\tau}$ and $\mathbf{\tau}^{\star}$ (Panel A for $K=4$), or a combination of 10 assets ($M=10$) plus portfolios $\mathbf{\tau}$ and $\mathbf{\tau}^{\star}$ (Panel B for $K=12$). We do so to gauge the testing performance  both in a small and a larger number of assets
to accommodate the empirical setting where we investigate spanning with up to 10 base assets.
To meet the conditions of Proposition \ref{optimal1}, we set the parameters of the multivariate GARCH process as $\mu_{i}=0$, for $i=1,...K-1$, while we choose $(a_i) = (0.4, 0.45, 0.5)$, $(\beta_i) = (0.5, 0.45, 0.4)$, $(\omega_{i})=(0.5,0.5,0.5)$, for $i=1,2,3$ (Panel A), and similarly
$(a_i) = (0.4, 0.41..., 0.5)$, $(\beta_i) = (0.5, 0.49,..., 0.4)$, $(\omega_{i})=(0.5,...,0.5)$, for $i=1,..,11$ (Panel B).
We set $v_1 = 1.5$ and $v_2 = 0.5$,  so that $\left\vert v_{1}\right\vert >\sqrt{\frac{\max\left\{ \omega_{i},a_{i},\beta_{i}\text{, }i=1,...,K-1\right\} }{\min\left\{ \omega_{i},a_{i},\beta_{i}\text{, }i=1,...,K-1\right\} }}$
and $\left\vert v_{2}\right\vert <\sqrt{\frac{\min\left\{ \omega_{i},a_{i},\beta_{i}\text{, }i=1,...,K-1\right\} }{\max\left\{ \omega_{i},a_{i},\beta_{i}\text{, }i=1,...,K-1\right\} }}$.
We use innovations generated by a Student distribution with 5 degrees of freedom.\footnote{Unreported simulation results for Gaussian innovations are similar.}

We use three different sample sizes. For $T=300$, we get the subsampling
distribution of the test statistic for subsample sizes $b_{T}\in\{50,100,150,200\}$.
We set $b_{T}\in\{100,200,300,400\}$ for $T=500$, and $b_{T}\in\{120,240,360,480\}$ for  $T=1000$.
We present the results using the original subsampling critical values (without bias correction) as well as 
the ones obtained using the bias correction method. The comparison shows that the bias correction improves a lot the inference in finite samples.  
The bias correction method eliminates the size distortion and delivers excellent properties under the alternative hypothesis with empirical powers above 90\% for a nominal size of 5\%.

In our simulations, the computational time is only marginally increasing with the number of assets, and is mainly increasing with the number of observations. For example, we have roughly 5 minutes for $T=300$ and the double for $T=500$ per run. Therefore we believe that the procedure can scale up to a couple of hundred assets. 

\begin{table}[ht]
\begin{centering}
\begin{tabular}{lcccccc}
\hline \hline 
\multicolumn{7}{c}{\textcolor{black}{Panel A: }\textbf{\textcolor{black}{$M=2, K=4$}}}\tabularnewline
\hline 
& \multicolumn{3}{c}{\textcolor{black}{Without bias correction }} & \multicolumn{3}{c}{\textcolor{black}{With bias correction }}
\tabularnewline
\hline 
& $T$=300 & $T$=500  & $T$=1000 & $T$=300 & $T$=500  & $T$=1000
\tabularnewline
\hline 
Size   & 12.6\% & 10.7\% & 8.2\% & 4.4\% & 3.6\% & 4.8\%\tabularnewline
Power    & 85.1\% & 87.4\% & 91.7\% & 93.7\% & 92.5\% & 96.2\% \tabularnewline
\hline 
\hline 
\multicolumn{7}{c}{\textcolor{black}{Panel B: }\textbf{\textcolor{black}{$M=10, K=12$}}}\tabularnewline
\hline 
& \multicolumn{3}{c}{\textcolor{black}{Without bias correction }} & \multicolumn{3}{c}{\textcolor{black}{With bias correction }}
\tabularnewline
\hline 
& $T$=300 & $T$=500  & $T$=1000  & $T$=300 & $T$=500  & $T$=1000 \tabularnewline
\hline 
Size   & 12.4\% & 10.5\% & 9.8\% & 5.6\% & 5.1\% & 4.4\% \tabularnewline
Power    & 85.7\% & 87.5\% & 89.1\%  & 92.1\% & 93.2\% & 95.8\% \tabularnewline
\hline 
\hline
\tabularnewline
\end{tabular}
\par\end{centering}
 \centering
\caption{Monte Carlo Results. Entries report the empirical size and empirical power based on 1000 replications, $T=300, 500, 1000$, and a nominal size $\alpha=5\%$. Panel A reports the rejection probabilities
for $M=2$ and $K=4$, while Panel B reports the rejection probabilities for $M=10$ and $K=12$. In both panels, we use a multivariate GARCH process to generate returns and compute the rejection probabilities without and with
the bias correction method for the subsampling critical values.} \label{tbl:sampling}
\end{table}

\section{Empirical Applications}

In this application, $\mathbb{L}$ consists of all convex combinations
of the market portfolio, the T-bill, and a set of base assets.
There is no need to explicitly allow for short selling in this application,
because the market portfolio has no binding short-sales restrictions;
non-binding constraints do not affect the efficiency classification.

Thanks to our spanning testing procedure, we want to check whether the two-fund separation theorem holds: can all
MSD investors combine the T-bill and the market portfolio to span
the whole set of their efficient portfolios?

If not, there is indication that active management for MSD investors
according to their preferences could outperform any combination of
the market portfolio and the riskless asset. This is studied in our second empirical application.

We use as base assets either the Fama and French (FF) size and book
to market portfolios, a set of momentum portfolios, a set of industry
portfolios, or a set of beta or size decile portfolios as described
below, along with the market portfolio and the T-bill. If the number of base assets
equals $n$, $\mathbb{L}$ is essentially the union of the relevant
$n-2$ subsimplex of the standard $n-1$ simplex with $\left\{ \left(0,\cdots,1\right)\right\} $,
where the latter signifies the market portfolio. The base assets, aside the market portfolio and 
the T-bill are the following portfolios:
\begin{itemize}
\item \textbf{The 6 FF benchmark portfolios}: They are constructed at the
end of each June, and correspond to the intersections of 2 portfolios
formed on size (market equity, ME) and 3 portfolios formed on the
ratio of book equity to market equity (BE/ME). 
\item \textbf{The 10 momentum portfolios}: They are constructed monthly
using NYSE prior (2-12) return decile breakpoints. The portfolios
include NYSE, AMEX, and NASDAQ stocks with prior return data. To be
included in a portfolio for month $t$ (formed at the end of month
$t-1$), a stock must have a price for the end of month $t-13$ and
a good return for $t-2$. 
\item \textbf{The 10 industry portfolios}: They are constructed by assigning
each NYSE, AMEX, and NASDAQ stock to an industry portfolio at the
end of June of year $t$ based on its four-digit SIC code at that
time. The industries are defined with the goal of having a manageable
number of distinct industries that cover all NYSE, AMEX, and NASDAQ
stocks. 
\item \textbf{The 10 size decile portfolios}: We use a standard set of ten
active US stock portfolios that are formed, and annually rebalanced,
based on individual stock market capitalization of equity (ME or size),
each representing a decile of the cross-section of NYSE, AMEX and
NASDAQ stocks in a given year. 
\item \textbf{The 10 beta decile portfolios}: We use a set of ten active
US stock portfolios that are formed, and annually rebalanced, based
on individual stock beta, each representing a decile of the cross-section
of NYSE, AMEX and NASDAQ stocks in a given year. 
\end{itemize}
For each dataset, we use data on monthly returns (month-end to month-end)
from January 1930 to December 2016 (1044 monthly observations) obtained
from the data library on the homepage\footnote{ http://mba.turc.dartmouth.edu/pages/faculty/ken.french}
of Kenneth French. The test portfolio is the Fama and French market
portfolio, which is the value-weighted average of all non-financial
common stocks listed on NYSE, AMEX, and Nasdaq, and covered by CRSP
and COMPUSTAT.

The portfolios used as base assets are of particular interest, because
a wealth of empirical research, starting with Banz (1981), Basu (1983),
and Fama and French (1993, 1997), suggests that the historical return
spread between small value stocks and small growth stocks defies rational
explanations based on investment risk. Moreover, book-to-market based
sorts are the basis for the factor model examined in Fama and French
(1993). Additionally, academics and practitioners show strong interest
in momentum portfolios. Empirical evidence indicates that common stocks
exhibit high returns on a period of 3-12 months will overperform on
subsequent periods. This momentum phenomenon is an important challenge
for the concept of market efficiency. Finally, industry sorted portfolios
have posed a particularly challenging feature from the perspective
of systematic risk measurement (see Fama and French (1997)). Beta-sorted
portfolios have been used extensively to test the Sharpe-Lintner Mossin
Capital Asset Pricing Model (CAPM) (see Black, Jensen, and Scholes
(1972), Blume and Friend (1973), Fama and MacBeth (1973), Reinganum
(1981), and Fama and French (1992), among others). Equity portfolios
have also been at the center of the empirical  literature
in the stochastic dominance framework, see for example Post (2003), Kuosmanen
(2004), Post and Levy (2005), Scaillet and Topaloglou (2010),  Post and Kopa (2013), Gonzalo and Olmo (2014), among others.

To focus on the role of preferences and beliefs, we adhere to the
assumptions of a single-period, portfolio-oriented model of a competitive
capital market. The model-free nature of SD tests seems an advantage
in this application area, because financial economists disagree about
the relevant shape of utility functions of investors and the probability
distribution of stock returns.

\subsection{Results of the MSD Spanning Test}

Arvanitis and Topaloglou
(2017) report evidence against the market portfolio being MSD efficient for data up to December 2012. 
We corroborate their findings in unreported results for our whole period up to December 2016 as well as two sub-periods, the first one
from January 1930 to June 1975, a total of 522 monthly observations,
and the second one from July 1975 to December 2016, 522 monthly observations.
Thus, we find evidence that passive
investment is suboptimal for investors with MSD preferences. Equity
management, instead of a standard buy-and-hold strategy on the market
portfolio, seems more appealing for investors with reverse S-shaped
utility functions. The MSD inefficiency of the market portfolio is not affected by transformations
that are increasing and convex over gains and increasing and concave
over losses, i.e., reverse S-shaped transformations.

Since the market is MSD inefficient, our next research
hypothesis is whether two-fund separation holds, i.e., whether all
MSD investors can satisfy themselves with combining the T-bill and the
market portfolio only. The  test of MSD efficiency for a given portfolio developped by Arvanitis and Topaloglou
(2017) cannot answer that question since their approach  is limited to the simple case of a spanning test for $\mathbb{K}$ being a singleton, and not any linear combination of two assets. 

For non-normal distributions, two-fund separation generally does not
occur, unless one assumes that preferences are sufficiently similar
across investors (see, for example, Cass and Stiglitz (1970)). Our
MSD spanning test can analyze two-fund separation without assuming
a particular form for the return distribution or utility functions.

We get the subsampling distribution of the test statistic for subsample
size\linebreak{}
$b_{T}\in[120,240,360,480]$. Using OLS regression on the empirical
quantiles $q_{T,b_{T}}(1-\alpha)$, for significance level $\alpha=0.05$,
we get the estimate $q_{T}$ for the critical value. We reject the
MSD spanning if the test statistic $\xi_{T}$ is higher than the
regression estimate $q_{T}$. In all the considered cases, $\mathbb{L}=\mathbb{S}$,
and $\alpha<\frac{3}{4}\leq1-ch_{\mathbb{L}}\left(\mathbb{K}\right)$
holds. Hence, if our assumption framework is valid, we expect that
asymptotic exactness holds. We find that: 
\begin{itemize}
\item \textbf{The 6 FF benchmark portfolios}: The regression estimate $q_{T}=15.74$
is lower than the value of the test statistic $\xi_{T}=26.78$. 
\item \textbf{The 10 momentum portfolios}: The regression estimate $q_{T}=19.42$
is lower than the value of the test statistic $\xi_{T}=41.55$. 
\item \textbf{The 10 industry portfolios}: The regression estimate $q_{T}=22.46$
is lower than the value of the test statistic $\xi_{T}=31.74$. 
\item \textbf{The 10 size decile portfolios}: The regression estimate $q_{T}=19.62$
is lower than the value of the test statistic $\xi_{T}=32.34$. 
\item \textbf{The 10 beta decile portfolios}: The regression estimate $q_{T}=31.48$
is lower than the value of the test statistic $\xi_{T}=44.76$. 
\end{itemize}
The results suggest the rejection of MSD spanning and thus of the two-fund
separation theorem for MSD investors. We get similar findings (unreported results) for the two subperiods 01/1930-06/1975 and 07/1975-12/ 2016.

As a final step in this analysis, we test for two-fund separation using
the Mean-Variance criterion rather than the MSD criterion. We use
the same methodology as for the above prospect spanning test, but
we restrict the utility functions to take a quadratic shape. We solve
the embedded expected-utility optimization problems (for every given
quadratic utility function) using quadratic programming. In contrast
to MSD spanning, we cannot reject the Mean-Variance spanning at conventional
significant levels.

The combined results of the market MSD efficiency and market MSD
spanning tests suggest that combining the T-bill and market portfolio
is not optimal for some MSD investors. Investors with reverse S-shaped
utility functions are investors that could outperform the market by
staying away from a buy-and-hold strategy on the market. Active investors
often take concentrated positions in assets with high upside potential
or follow dynamic strategies like momentum. They can also prefer looking at defensive strategies. 
That can produce opportunities 
with positively skewed returns, or at least less 
negatively skewed, which are attractive for MSD investors.

\subsection{Performance Summary of the MSD portfolios}

The rejection of the spanning hypothesis implies that there exists at least one portfolio in $\mathbb{L}$ 
which is weakly prefered to every portfolio in $\mathbb{K}$ by at least one reverse S-shaped utility function 
(see Definition \ref{M-span}). Such a portfolio is by construction efficient w.r.t.\ $\mathbb{K}$ 
(see Definition 2.1 in Linton et al.\ (2014) for the SSD case which can be easily generalized 
to our MSD case). The empirical version of such a portfolio is the optimal portfolio $\lambda$ that 
maximizes $\xi_{T}$ for the particular sample value.  In what follows, and given this characterization, we 
analyze the performance of such empirically optimal MSD portfolios through time, compared to the 
performance of the market portfolio (buy-and-hold strategy).

We resort to backtesting experiments on a rolling window basis. The
rolling horizon computations cover the 642-month period from 07/1963
to 12/2016. At each month, we use the data from the previous 30 years
(360 monthly observations) to calibrate the procedure. We solve the
resulting optimization model for the MSD spanning test and record
the optimal portfolio made of the base assets as well as the market
portfolio and the T-bill. We determine the realized return
of the chosen MSD optimal portfolio from the actual returns of
the  asset weight allocation picked by the optimizer for that month. Then, we repeat the same procedure for the next
one-month rolling window and compute  the ex-post realized returns for the period from
07/1963 to 12/2016. Therefore, the MSD optimal portfolios are outcomes of the testing procedure based on an unconditional distribution updated for each rolling window
and performance is realized out of the optimization sample  (no look-ahead bias).

Let us first compute the cumulative performance of the
MSD optimal portfolios as well as the market portfolio for the entire
sample period from July 1963 to December 2016 based on the optimal portfolio weights obtained for each one-month rolling window.
The value for the MSD optimal portfolios is 426 times higher at the end
of the holding period compared to the initial value, while the
market portfolio is only 13.9 times higher. Hence,
the relative performance of MSD type investors is 30 times higher
than the performance of the market in the evaluated period. Such an increase of 3000\% is significant at any significance level (unreported results).

To get further insights of the differences between two investment strategies,   we report the first four moments of the realized returns and the Value-at-Risk in Table 2. We further compute a number of commonly used performance measures:
the Sharpe ratio, the downside Sharpe ratio, the return loss and the
opportunity cost. 

The downside
Sharpe ratio based on the semi-variance (Ziemba (2005)) is considered to be a more appropriate measure
of performance than the typical Sharpe ratio given the asymmetric
return distribution of the assets.

To account for transaction costs,  we use the proposal of DeMiguel et al. (2009). This indicates the way that the proportional
transaction costs, generated by the portfolio turnover, affect the
portfolio returns. Let $trc$ be the proportional transaction cost,
and $R_{P,t+1}$ the realized return of portfolio $P$ at time $t+1$.
The change in the net of transaction cost wealth $NW_{P}$ of portfolio
$P$ through time is,
\begin{equation}
NW_{P,t+1}=NW_{P,t}(1+R_{P,t+1})[1-trc\times\sum_{i=1}^{N}(|w_{P,i,t+1}-w_{P,i,t}|).
\end{equation}
The portfolio return, net of transaction costs is defined as
\begin{equation}
RTC_{P,t+1}=\frac{NW_{P,t+1}}{NW_{P,t}}-1. \label{RTC}
\end{equation}
Let $\mu_{M}$ and $\mu_{MSD}$ be the out-of-sample mean of (\ref{RTC})
for the market portfolio and the MSD optimal portfolios, and $\sigma_{M}$ and $\sigma_{MSD}$ be the corresponding standard
deviations. Then, the return-loss measure is,
\begin{equation}
R_{Loss}=\frac{\mu_{MSD}}{\sigma_{MSD}}\times\sigma_{M}-\mu_{M},
\end{equation}
i.e., the additional return needed so that the market
performs equally well with the MSD optimal portfolios. We follow the
literature and use 35 bps for the transaction costs of stocks and bonds.

Finally, the opportunity cost presented in Simaan (2013)
gauges the economic significance of the performance difference
between two portfolios. Let $R_{MSD}$ and $R_{M}$ be the
realized returns of the MSD optimal portfolios and the
market portfolio, respectively. Then, the opportunity cost $\theta$
is defined as the return that needs to be added to (or subtracted
from) the market return $R_{M}$, so that the investor
is indifferent (in utility terms) between the strategies imposed by
the two different investment opportunity sets, i.e.,
\begin{equation}
E[U(1+R_{M}+\theta)]=E[U(1+R_{MSD})].
\end{equation}
A positive (negative) opportunity cost implies that the investor is
better (worse) off if the investment opportunity set allows for MSD type
investing. The opportunity cost takes into account the
entire probability density function of asset returns and hence it
is suitable to evaluate strategies even when the distribution
is not normal. For the calculation of the opportunity cost, we use
the following utility function which satisfies the curvature of Markowitz theory (reverse-S-shaped):
\begin{equation}
U(R)=
 \left\{
\begin{array}{l}
R^{a}, \quad if \mbox\ R \geq 0, \\ 
- c(-R)^{b}, \quad if  \mbox\ R <0,
\end{array}  
\right.
\end{equation}
where  $c$ is the coefficient of loss aversion (usually $c=2.25$) and $a, b > 1$. We use several values of $a,b$ in Table 2 to drive the curvature of the utility functions.

\begin{table}[ht]
\begin{tabular}{lcc} \hline \hline 
 &  {\small{}{}{}{}{}{}{}{}{}{}{}MSD optimal portfolio }    & {\small{}{}{}{}{}{}{}{}{}{}{}market portfolio}   \tabularnewline
\hline
{\small{}{}{}{}{}{}{}{}{}{}{}Mean }  & {\small{}{}{}{}{}{}{}{}{}{}{}0.01035 }  & {\small{}{}{}{}{}{}{}{}{}{}{}0.00510 }   \tabularnewline
{\small{}{}{}{}{}{}{}{}{}{}{}Standard Deviation }  & {\small{}{}{}{}{}{}{}{}{}{}{}0.04290 }  & {\small{}{}{}{}{}{}{}{}{}{}{}0.04420 }   \tabularnewline
{\small{}{}{}{}{}{}{}{}{}{}{}Skewness }  & {\small{}{}{}{}{}{}{}{}{}{}{}-0.27730 }  & {\small{}{}{}{}{}{}{}{}{}{}{}-0.52629 }   \tabularnewline
{\small{}{}{}{}{}{}{}{}{}{}{}Excess Kurtosis } & {\small{}{}{}{}{}{}{}{}{}{}{}1.18535 }  & {\small{}{}{}{}{}{}{}{}{}{}{}1.96705}   \tabularnewline
{\small{}{}{}{}{}{}{}{}{}{}{}VaR 5\% }  &  {\small{}{}{}{}{}{}{}{}{}{}{}0.06133}  &   {\small{}{}{}{}{}{}{}{}{}{}{}0.0718} \tabularnewline
{\small{}{}{}{}{}{}{}{}{}{}{}Sharpe Ratio } & {\small{}{}{}{}{}{}{}{}{}{}{}0.17495 }  & {\small{}{}{}{}{}{}{}{}{}{}{}0.04697 }   \tabularnewline
{\small{}{}{}{}{}{}{}{}{}{}{}{}{}Downside Sharpe Ratio}  & {\small{}{}{}{}{}{}{}{}{}{}{}{}{}0.12986}  & {\small{}{}{}{}{}{}{}{}{}{}{}{}{}0.39570}   \tabularnewline
{\small{}{}{}{}{}{}{}{}{}{}{}{}{}Return Loss}  & {\small{}{}{}{}{}{}{}{}{}{}{}{}{}0.7856\%}  &   \tabularnewline
{\small{}{}{}{}{}{}{}{}{}{}{}{}{}Opportunity Cost  ($c=2.25$)}  &  &  \tabularnewline   
{\small{}{}{}{}{}{}{}{}{}{}{}{}{}   $a =b = 2$}  & {\small{}{}{}{}{}{}{}{}{}{}{}{}{}0.704\%}  &   \tabularnewline
{\small{}{}{}{}{}{}{}{}{}{}{}{}{}  $a =b= 3$}  & {\small{}{}{}{}{}{}{}{}{}{}{}{}{}0.990\%}  &  \tabularnewline
{\small{}{}{}{}{}{}{}{}{}{}{}{}{}   $a =b = 4$}  & {\small{}{}{}{}{}{}{}{}{}{}{}{}{}1.565\%}  &   \tabularnewline
\hline\hline 
 \tabularnewline
\end{tabular}
\centering{}\caption{Performance and risk measures. Entries report 
performance and risk measures  for the MSD optimal portfolios
and the market portfolio computed with one-month rolling windows. 
We list mean, volatility, skewness, excess kurtosis, empirical VaR 5\% (positive sign for a  loss), Sharpe ratio, downside Sharpe ratio, 
return loss, and opportunity cost.
The dataset spans the period from July 31, 1963 to December 31, 2016.} 
\label{tbl:t4}
\end{table}

Table \ref{tbl:t4} reports the performance and risk measures for the MSD optimal portfolios
and the market portfolio. 
These measures allow us to better figure out the differences between the market portfolio and the MSD strategy. The mean is higher for the MSD optimal 
portfolio and the variance is lower, which results in a higher Sharpe ratio.
The skewness is less negative as expected for a portfolio built for investors with 
preferences towards risk that are associated with risk aversion for losses and risk loving for gains.
The kurtosis and VaR are lower as expected when investors want to mitigate the 
impact of large losses.
The MSD portfolio targets and achieves a transfer of probability mass from the 
left to the right tail of the return distribution when compared to the market portfolio. The opportunity cost is above 70 bps  
and increases with the curvatures of the gain and loss parts of the utility function.  

Table \ref{tbl:t4weights} reports the descriptive statistics regarding the weight allocation of the MSD optimal portfolios.
They load mainly on big size FF portfolios (FF portfolios), several momentum portfolios (momentum portfolios),
telecommunications, health, energy and utilities (industry portfolios),
small caps (size portfolios ), and  low and medium
beta (beta sorted portfolios), in addition to the market portfolio
and the T-Bill.

\begin{table}[ht]
\begin{center}
\begin{tabular}{lccccc} \hline \hline 
\multicolumn{6}{l} {Descriptive statistics of the weight allocation of the MSD portfolio} \\\hline
  Base Assets & Portfolio & Mean & Std. Dev. & Skewness & Kurtosis  \\\hline
  & Market  	&0.0933	& 0.0461	& -0.5990	& 0.0279 \\
   & T-Bill  		&0.0265	& 0.0678	& 2.1743	& 2.7361 \\\hline
  6 FF &Big LoBM 	&0.0282  	& 0.0472 	& 1.7689 	& 2.4370 \\
    	& Big AvBM	&0.0203  	& 0.0529	& 2.2407	 & 3.0765  \\
  	 & Big Hi BM &0.1515 	& 0.06231 &  -0.9441	 & -0.6088  \\\hline
  10 momentum & Prior 3	& 0.0539   & 0.0397	 & -0.5491	 & -1.6188  \\
   	& Prior 4	& 0.0178   & 0.0329 	 & 2.0128	 & 3.2530 \\
   	& Prior 5	& 0.0249   & 0.0554	 & 1.7780	 & 1.1651  \\
   	& Prior 7	& 0.023	  & 0.0099	 &  4.0016	 & 14.056  \\
   	& Prior 8 	&0.0097	 & 0.0165	  & 1.1428	 & -0.6570  \\
	& Prior 9	&0.0027	  & 0.0101	 & 3.5233	 & 10.446\\
   	& Hi Prior	& 0.0886	  & 0.0525	 & -0.7163	 & -0.3800    \\\hline
  10 industry 	&Telcm 	&0.0301	  & 0.0464	 &  1.2129	 & 0.1298  \\
   	& Hith	&0.0327	  & 0.0680	 &  1.8638	 & 1.8056 \\
  	& Utils	&0.0124	 & 0.0349	 & 4.3510	 & 18.234  \\
   	& Energy	&0.0984	  & 0.0717	 & -0.2824 & -1.4674  \\\hline
  10 size & 		ME1		&0.1911	  & 0.0156	 & -1.1916	 & -0.5817  \\\hline
  10 beta	 	&Lo 10 	&0.0671	  & 0.0519	 & -0.3527	 & -1.7449   \\
  	&Quant. 20 &0.0318	  & 0.0627	 & 1.9163 	& 2.1448  \\ 
	& Quant. 30&0.0052	  &0.0198	 & 3.5233	 & 10.446  \\
 	&Quant. 40 &0.0025	  & 0.0106	 &  4.0016 	& 14.056 \\\hline\hline 
\end{tabular}
\end{center}
\caption{Descriptive statistics of the weight allocation of the MSD optimal portfolios over the period
07/1963-12/2016 computed with one-month rolling windows.} \label{tbl:t4weights}
\end{table}

We also investigate which factors explain the returns of the active
investors with MSD preferences. To do so, we use the four-factor model of Carhart (1997) which adds momentum in the three-factor model of Fama and French (1992, 1993),
as well as the Fama and French five-factor model (2015). Our empirical test examines whether these models explain the returns on MSD portfolios that dominate
any combination of the market and the riskless asset, namely whether standard factors used in the empirical asset pricing literature are potential drivers of returns of MSD optimal portfolios.

First, we consider the following linear regression (Carhart four-factor model):
\begin{equation}
R_{it}-R_{Ft}=a_{i}+b_{i}(R_{Mt}-R_{Ft})+s_{i}SMB_{t}+h_{i}HML_{t}+r_{i}MOM_{t}+e_{it}, \nonumber
\end{equation}
where $R_{it}$ is the return of the MSD optimal portfolio at period
$t$, $R_{Ft}$ is the riskless rate, $R_{Mt}$ is the return on the
value-weight (VW) market portfolio, $SMB_{t}$ is the return on a
diversified portfolio of small stocks minus the return on a diversified
portfolio of big stocks, $HML_{t}$ is the difference between the returns
on diversified portfolios of high and low BE/ME stocks, $MOM_{t}$
is the average return on the two high prior return portfolios minus the 
average return on the two low prior return portfolios, 
and $e_{it}$ is a zero-mean residual. If the exposures $b_{i}$,
$s_{i}$, $h_{i}$, and $r_{i}$ to the market, size, value,
and momentum factors capture all variation in expected
returns, the intercept $a_{i}$ is zero.

\begin{table}[ht]
\begin{tabular}{lrrrrr}
\hline \hline 
 & $a_{i}$  & $R_{M}-R_{F}$  & $SMB$  & $HML$  & $MOM$  \tabularnewline
\hline 
Coef.  & 0.508  & 0.948  & -0.031  & 0.133  & 0.004   \tabularnewline
$t$-stat  & 1.294  & 1.021  & -2.484  & 9,380  & 0.441  \tabularnewline
$p$-values  & 0  & 0  & 0.013  & 0  & 0.659   \tabularnewline
\hline \hline 
\end{tabular}\medskip{}
\begin{tabular}{lrr}
\hline \hline 
Adj. $R^{2}$  & F-statistic  & $p$-value \tabularnewline
\hline 
0.948  & 2.953  & 0\tabularnewline
\hline \hline 
 \tabularnewline
\end{tabular}\centering{}\caption{Carhart four-factor model. Entries report the coefficients and their respective $t$-statistics,as well as Adjusted R2, F-statistic, and
 $p$-values. The dataset spans 07/1963-12/2016for MSD optimal portfolios computed with one-month rolling windows. }
\label{tbl:statistics1} 
\end{table}

Table~\ref{tbl:statistics1} reports the coefficient estimates of
the four factors, as well as their respective $t$-statistics and $p$-values.
The results indicate that apart from the momentum ($MOM$), all
the other three factors explain part of the performance of the optimal MSD
portfolios. The intercept is not zero, which indicates that perhaps
other factors drive the performance of the MSD portfolios as well.

We additionally consider the following linear regression (five-factor model):
\begin{equation}
R_{it}-R_{Ft}=a_{i}+b_{i}(R_{Mt}-R_{Ft})+s_{i}SMB_{t}+h_{i}HML_{t}+r_{i}RMW_{t}+c_{i}CMA_{t}+e_{it}, \nonumber
\end{equation}
where $R_{it}$ is the return of the MSD optimal portfolio at period
$t$, $R_{Ft}$, $R_{Mt}$, $SMB_{t}$ and $HML_{t}$ as before, $RMW_{t}$
is the difference between the returns on diversified portfolios of
stocks with robust and weak profitability, $CMA_{t}$ is the difference
between the returns on diversified portfolios of the stocks of low
and high investment firms, which are called conservative and aggressive,
and $e_{it}$ is a zero-mean residual. If the exposures $b_{i}$,
$s_{i}$, $h_{i}$, $r_{i}$, and $c_{i}$ to the market, size, value,
profitability and investment factors capture all variation in expected
returns, the intercept $a_{i}$ is zero.

\begin{table}[ht]
\begin{tabular}{lrrrrrr}
\hline \hline 
 & $a_{i}$  & $R_{M}-R_{F}$  & $SMB$  & $HML$  & $RMW$  & $CMA$ \tabularnewline
\hline 
Coef.  & 0.419  & 0.981  & -0.019  & 0.201  & 0.021  & -0.06 \tabularnewline
$t$-stat  & 15.30  & 146.3  & -2.075  & 15.51  & 1.597  & -3.327 \tabularnewline
$p$-values  & 0  & 0  & 0.038  & 0  & 0.111  & 0.009 \tabularnewline
\hline \hline 
\end{tabular}\medskip{}
\begin{tabular}{lrr}
\hline 
Adj. $R^{2}$  & F-statistic  & $p$-value \tabularnewline
\hline 
0.988  & 5361.6  & 0\tabularnewline
\hline \hline 
 \tabularnewline
\end{tabular}\centering{}\caption{Fama-French five factor model. Entries report the coefficient estimates,
their respective $t$-statistics, as well as Adjusted R2, F-statistic, and
$p$-values. The dataset spans 07/1963-12/2016 for MSD optimal portfolios computed with one-month rolling windows. }
\label{tbl:statistics2} 
\end{table}

Table~\ref{tbl:statistics2} reports the coefficient estimates of
the five-factor model, as well as their respective $t$-statistics and $p$-values.
The results indicate that, apart from the profitability ($RMW$), all
the other four factors explain part of the performance of the optimal MSD
portfolios. The intercept clearly not being zero indicates that 
other factors possibly drive the performance of the MSD portfolios as well.

In both factor models, we observe that the beta market is slightly smaller than one (defensive) for the MSD portfolios as expected.
The negative sign for the SMB factor loading and positive sign for the HML factor loading correspond to an additional defensive tilt. 
Defensive strategies overweight large value stocks and underweight small growth stocks (see Novy-Marx (2016)).

\section{Conclusions}

We have derived properties of the cdf of a random variable defined
by recursive optimizations applied on a continuous stochastic process
w.r.t.\ possibly dependent parameter spaces. Those properties extend
previous results and can be useful for the derivation of the limit
theory of tests for stochastic spanning w.r.t.\ stochastic dominance
relations.

As a theoretical application, we have defined the concept of spanning,
constructed an analogous  test based on subsampling, and derived the
first-order limit theory and a numerical implementation for the case
of the MSD relation.

We have used the non-parametric test in an empirical application, inspired by Arvanitis
and Topaloglou (2017),  who show that the market portfolio is not MSD
efficient. The spanning test enables us to explore whether MSD equity
managers could outperform the market portfolio. First, we test
whether the market portfolio is MSD efficient, and then whether the
two-fund separation theorem holds for investors with MSD preferences.
We use as base assets either the FF size and book to market portfolios,
a set of momentum portfolios, a set of industry portfolios, or a set
of beta or size decile portfolios. Empirical results indicate that
the market portfolio is not MSD efficient, and the two-fund separation
theorem does not hold for MSD investors. Thus, the combination of
the market and the riskless asset do not span the portfolios created
according to the MSD criterion. Hence, there exist MSD investors that
could benefit from investment opportunities that involve assets beyond
portfolios constructed solely by the market portfolio and the safe
asset. We verify this by showing that equity managers with MSD preferences
could generate portfolios that yield 30 times higher cumulative return
than the market over the last 50 years. The return distribution of the MSD optimal 
portfolio is less negatively skewed, less leptokurtic, and thiner left-tailed,  
when compared to the market portfolio. Finally, using the four-factor  model of Carhart (1997) and  the  five-factor model of 
Fama and French (2015), we investigate which factors explain these returns. We find that a defensive tilt explains part 
of the performance of the optimal MSD portfolios, while momentum and profitability do not.

The derivations and methodology used above can also be explored for
other forms of stochastic dominance relations, such as the first- or
the third-order, or Prospect stochastic dominance. We leave such issues
for future research.

\section{Appendix}

\subsection{Proofs of Main Results}
\begin{proof}[Proof of Theorem \ref{aac}]
First, we know that $\xi\in\mathbb{D}^{1,2}$, from similar arguments
to the ones in the proof of Proposition 2.1.10 of Nualart (2006).
Precisely, consider a countable dense subset of $\Lambda$ , say $\Lambda_{\infty}$
as well as $\xi_{n}:=\text{oper}X_{\lambda}$, where $\mbox{opt}_{i}$
is considered w.r.t.\ $\Lambda_{i,n}^{\star}\left(\lambda_{i-1}\right)=\left\{ \textnormal{the first \ensuremath{n} elements of }\Lambda_{i}^{\star}\left(\lambda_{i-1}\right)\cap\textnormal{pr}_{i}\Lambda_{\infty}\right\} $
and $\lambda_{i-1}\in\Lambda_{i-1,n}^{\star}$ when $i>1$. The function
$\text{oper}:C\left(\Lambda,\mathbb{\mathbb{R}}\right)\rightarrow\mathbb{R}$
is Lipschitz, and from Proposition 1.2.4 of Nualart (2006), we get
$\eta_{n}\in\mathbb{D}^{1,2}$. Furthermore, from Assumption \ref{as1}.1,
$\xi_{n}\rightarrow\xi$ in $L^{2}\left(\Omega\right)$, and therefore
the preliminary result follows if $\left(D\xi_{n}\right)_{n\in\mathbb{N}}$
is $L^{2}\left(\Omega\right)$ bounded. Define 
\[
A_{n}=\left\{ \omega\in\Omega:\xi_{n}=X_{\lambda_{n}},\xi_{n}\neq X_{\lambda_{k}},\forall k<n\right\} .
\]
Using the local property of $D$, we have that $D\xi_{n}=\sum_{n\in\mathbb{N}}1_{A_{n}}DX_{\lambda_{n}},$
and thereby $\mathbb{E}\left[\mathbb{\|}D\xi_{n}\|_{H}^{2}\right]<+\infty$
from Assumption \ref{as1}.2. Then Assumption \ref{as1}.3 as well
as Proposition 2.1.7 of Nualart (2006) imply the first part of the
theorem. For the following, assume first that $T$ is empty. Then
the result will follow from a series of arguments almost identical
to the ones in the proof of Proposition 2.1.11 of Nualart (2006).
Specifically, consider the set 
\[
G=\left\{ \omega\in\Omega:\mbox{ there exists }\lambda\in\Lambda\mbox{ such that \ensuremath{DX_{\lambda}\neq D\xi\mbox{ and }}\ensuremath{X_{\lambda}=\xi}}\right\} ,
\]
and using $\Lambda_{\infty}$ above $H_{\infty}$ a countable dense
subset of the unit ball of $H$, and $B_{r}\left(\lambda\right)$
the ball in $\Lambda$ with center $\lambda$ and radius $r>0$ we
have that $G\subseteq\cup_{\lambda\in\Lambda_{\infty},r\in\mathbb{Q}_{++},k\in\mathbb{N}_{0},h\in H_{\infty}}G_{\lambda,r,k,h}$
i.e., a countable union, where 
\[
G_{\lambda,r,k,h}:=\left\{ \omega\in\Omega:\left\langle DX_{\lambda'}-D\eta,h\right\rangle >\frac{1}{k}\mbox{ for all \ensuremath{\lambda'\in B_{r}\left(\lambda\right)}}\right\} \cap\left\{ \text{oper}X_{\lambda'}=\xi\right\} .
\]
For some $\lambda,r,k,h$ as above, define $\xi'=\mbox{oper}X_{\lambda'}$,
where now $\mbox{opt}_{i}$ is considered w.r.t.\ $\Lambda_{i}^{\star}\left(\lambda_{i-1}\right)\cap\textnormal{pr}_{1}B_{r}\left(\lambda\right)$
choose a countable dense subset of $B_{r}\left(\lambda\right)$, say
$B_{r}^{\infty}\left(\lambda\right)$ and using 
\[
\Lambda_{i,n}^{\infty}\left(\lambda_{i-1}\right)=\left\{ \textnormal{the first \ensuremath{n} elements of }\Lambda_{i}^{\star}\left(\lambda_{i-1}\right)\cap\textnormal{pr}_{i}B_{r}^{\infty}\left(\lambda\right)\right\} ,
\]
define $\xi_{n}^{'}=\text{oper}X_{\lambda}$ analogously. We have
that as $n\rightarrow\infty$ $\xi_{n}^{'}\rightarrow\xi'$ in $L^{2}\left(\Omega\right)$
norm due to Assumption \ref{as1}.1. From Lemma 1.2.3 of Nualart (2006)
and Assumption \ref{as1}.2 we also have that $D\xi_{n}^{'}\rightarrow D\xi'$
in the weak topology of $L^{2}\left(\Omega,H\right)$. Using again
the local property argument as above, we have that for any $\omega\in G_{\lambda,r,k,h}$,
$D\xi_{n}^{'}=DX_{\lambda'}$, for some $\lambda'\in B_{r}^{\infty}\left(\lambda\right)$.
But, for such $\omega$, we have that $\left\langle D\xi_{n}^{'}-D\xi^{'},h\right\rangle >\frac{1}{k}$
for all $n$. This directly implies that $\mathbb{P}\left(G_{\lambda,r,k,h}\right)=0$
which, due to countability, implies that $\mathbb{P}\left(G\right)=0$.
Then the result follows from Theorem 2.1.3 of Nualart (2006). Now,
suppose that $\tau\in\mathcal{T}$, and consider 
\[
\mathbb{P}\left(\xi=\tau\right)=\mathbb{P}\left(\left\{ \xi=\tau\right\} \cap\text{\ensuremath{\Omega}}_{\tau}\right)+\mathbb{P}\left(\left\{ \xi=\tau\right\} \cap\text{\ensuremath{\Omega}}_{\tau}^{c}\right)
\]
If, for some $\tau\in\mathcal{T}$, $\mathbb{P}\left(\Omega_{\tau}^{c}\right)>0$,
we get 
\[
\mathbb{P}\left(\left\{ \xi=\tau\right\} \cap\Omega_{\tau}^{c}\right)=\mathbb{P}\left(\xi=\tau/\Omega_{\tau}^{c}\right)\mathbb{P}\left(\Omega_{\tau}^{c}\right),
\]
and we can consider the process $X^{\star}:=X\left|_{\Omega-\cup_{\tau\in\mathcal{T}}\Omega_{\tau}^{c}}\right.$
that obviously satisfies Assumption \ref{as1} with $\mathcal{T}^{\star}=\emptyset$
along with the obvious change of notation. Hence, $\xi^{\star}$ has
an absolutely continuous law something that implies that $\mathbb{P}\left(\xi=\tau/\Omega_{\tau}^{c}\right)=\mathbb{P}\left(\xi^{\star}=\tau\right)=0$.
If $\mathbb{P}\left(\Omega_{\tau}^{c}\right)=0$ trivially $\mathbb{P}\left(\left\{ \xi=\tau\right\} \cap\Omega_{\tau}^{c}\right)=0$
establishing that $\mathbb{P}\left(\xi=\tau\right)=\mathbb{P}\left(\left\{ \xi=\tau\right\} \cap\Omega_{\tau}^{c}\right)$
in any case. Now, suppose that $\tau_{1},\tau_{2}$ are successive
elements of $\mathcal{T}$ and consider $\text{\ensuremath{\Omega}}_{\tau_{1},\tau_{2}}=\left\{ \omega\in\text{\ensuremath{\Omega}}:\xi\in\left(\tau_{1},\tau_{2}\right)\right\} $.
The previous imply that $\mathbb{P}\left(\text{\ensuremath{\Omega}}_{\tau_{1},\tau_{2}}\right)>0$,
hence the process $X_{\star}:=X\left|_{\text{\ensuremath{\Omega}}_{\tau_{1},\tau_{2}}}\right.$
satisfies Assumption \ref{as1} with $\mathcal{T}_{\star}=\emptyset$,
and thereby $\xi_{\star}$ has an absolutely continuous law. The other
cases follow analogously when the intersections apperaring in the
theorem are non empty. When empty the results are trivial. 
\end{proof}
\begin{proof}[Proof of Corollary \ref{cor:used}]
It follows simply by Theorem \ref{aac} since the relation between
$\xi$ and $\eta$ implies that $\text{supp}\left(\xi\right)$ is
the closure of $\left(c,+\infty\right)$ and also that $\mathbb{P}\left(\xi=c\right)\leq\mathbb{P}\left(\eta=c\right)$. 
\end{proof}
\begin{proof}[Proof of Proposition \ref{M-equiv}]
($\Leftarrow$) If $\mathbb{K}\succcurlyeq_{M}\mathbb{L}$, for any
$\lambda$, there exists some $\kappa$ such that\linebreak{}
 $\sup_{z\leq0}\Delta_{1}\left(z,\lambda,\kappa,F\right)\leq0$ and
$\sup_{z>0}\Delta_{2}\left(z,\lambda,\kappa,F\right)\leq0$. This
implies that 
\begin{equation}
\max_{i=1,2}\sup_{z\in A_{i}}\inf_{\kappa\in\mathbb{K}}\Delta_{i}\left(z,\lambda,\kappa,F\right)\leq0.\label{eq:3}
\end{equation}
Since $\mathbb{K}$ is closed, hence compact, and $F$ has a finite
first moment, the Dominated Convergence Theorem implies that $\mathcal{J}\left(-\infty,0,\kappa,F\right)$
is continuous w.r.t.\ $\kappa$. This along with the compactness of
$\mathbb{K}$ imply that $\arg\min_{\kappa\in\mathbb{K}}\mathcal{J}\left(-\infty,0,\kappa,F\right)$
is non empty. Let $\kappa^{\star}$ be an element of the latter. Then,
the first equality follows from 
\[
\xi\left(F\right)\geq\inf_{\kappa\in\mathbb{K}}\mathcal{J}\left(-\infty,0,\kappa,F\right)-\mathcal{J}\left(-\infty,0,\kappa^{\star},F\right)=0.
\]
If $\mathbb{K}\nsucceq_{M}\mathbb{L}$ for some $\lambda^{\star}\in\mathbb{L},$
and any $\kappa\in\mathbb{K}$, there exists some $i\left(\lambda^{\star},\kappa\right),z^{\star}\left(\lambda^{\star},\kappa\right)\in A_{i}$
such that $\Delta_{i}\left(z,\lambda,\kappa,F\right)>0$. Then the
continuity of $\mathcal{J}\left(-\infty,z,\kappa,F\right)$ and $\mathcal{J}\left(z,+\infty,\kappa,F\right)$
w.r.t.\ $\kappa$, and the compactness of $\mathbb{K}$, imply that,
for any $\lambda\notin\mathbb{K},z\in A_{1},\exists\kappa_{\lambda,z}\in\mathbb{K}$
such that 
\[
\inf_{\kappa\in\mathbb{K}}\Delta_{1}\left(z,\lambda,\kappa,F\right)=\Delta_{1}\left(z,\lambda,\kappa_{\lambda,z},F\right),
\]
and thereby 
\[
\xi\left(F\right)\geq\Delta_{1\left(\lambda^{\star},\kappa_{\lambda^{\star},z^{\star}}\right)}\left(z^{\star},\lambda^{\star},\kappa_{\lambda^{\star},z^{\star}},F\right)>0.
\]
($\Rightarrow$) Suppose now that $\xi\left(F\right)=0$ and consider
an arbitrary $\lambda$. This implies that (\ref{eq:3}) holds and
thereby there exists some element of $\mathbb{K}$ for which $\Delta_{i}\left(z,\lambda,\kappa,F\right)\leq0$,
for every $z\in A_{i},\:i=1,2$. If $\xi\left(F\right)>0$, for some
$\lambda^{\star}\in\mathbb{L}$, and some $i=1,2$, $\inf_{\kappa\in\mathbb{K}}\sup_{z\in A_{i}}\Delta_{i}\left(z,\lambda^{\star},\kappa,F\right)>0.$
It implies that for any $\kappa\in\mathbb{K}$, $\sup_{z\in A_{i}}\Delta_{i}\left(z,\lambda^{\star},\kappa,F\right)>0$
and the result follows. 
\end{proof}
\begin{proof}[Proof of Proposition \ref{EAD}]
The results in the auxiliary Lemma \ref{lem:emp_pr} imply that\linebreak{}
 $\left(\begin{array}{c}
\Delta_{1}\left(z_{1},\lambda,\kappa,\sqrt{T}\left(F_{T}-F\right)\right)\\
\Delta_{2}\left(z_{2},\lambda,\kappa,\sqrt{T}\left(F_{T}-F\right)\right)
\end{array}\right)$ weakly converges to $\left(\begin{array}{c}
\Delta_{1}\left(z_{1},\lambda,\kappa,\mathcal{G}_{F}\right)\\
\Delta_{2}\left(z_{2},\lambda,\kappa,\mathcal{G}_{F}\right)
\end{array}\right)$ w.r.t.\ to the product topology of continuous (w.r.t.\ $\left(z_{1},z_{2},\lambda\right)$)
epi-convergence (w.r.t.\ $\kappa$) on the product of the relevant
spaces of lsc real valued functions (see e.g. Knight (1999) for the
dual notion of epi-convergence). This product space is metrizable
as complete and separable (see again Knight (1999)). Hence, Skorokhod
representations are applicable (as above, see for example Theorem
1 in Cortissoz (2007)) and thereby for any $\left(z_{1},z_{2},\lambda\right)$
and any sequence $\left(z_{1,T},z_{2,T},\lambda_{T}\right)\rightarrow\left(z_{1},z_{2},\lambda\right)$,
there exists an enhanced probability space and processes\linebreak{}
 $\left(\begin{array}{c}
\Delta_{1,T}\left(\kappa\right)\\
\Delta_{2,T}\left(\kappa\right)
\end{array}\right)\overset{d}{=}\left(\begin{array}{c}
\Delta_{1}\left(z_{1,T},\lambda_{T},\kappa,\sqrt{T}\left(F_{T}-F\right)\right)\\
\Delta_{2}\left(z_{2,T},\lambda_{T},\kappa,\sqrt{T}\left(F_{T}-F\right)\right)
\end{array}\right)$, $\left(\begin{array}{c}
\Delta_{1}^{\star}\left(\kappa\right)\\
\Delta_{2}^{\star}\left(\kappa\right)
\end{array}\right)\overset{d}{=}\left(\begin{array}{c}
\Delta_{1}\left(z_{1},\lambda,\kappa,\mathcal{G}_{F}\right)\\
\Delta_{2}\left(z_{2},\lambda,\kappa,\mathcal{G}_{F}\right)
\end{array}\right)$, defined on it such that $\left(\begin{array}{c}
\Delta_{1,T}\\
\Delta_{2,T}
\end{array}\right)\rightarrow\left(\begin{array}{c}
\Delta_{1}^{\star}\\
\Delta_{2}^{\star}
\end{array}\right)$ almost surely, w.r.t.\ to the product topology of epi-convergence,
where $\overset{d}{=}$ denotes equality in distribution. Notice that,
\[
\left(\begin{array}{c}
\Delta_{1}\left(z_{1,T},\lambda_{T},\kappa,\sqrt{T}F_{T}\right)\\
\Delta_{2}\left(z_{2,T},\lambda_{T},\kappa,\sqrt{T}F_{T}\right)
\end{array}\right)\overset{d}{=}\left(\begin{array}{c}
K_{1,T}\left(\kappa\right)\\
K_{2,T}\left(\kappa\right)
\end{array}\right):=\left(\begin{array}{c}
\Delta_{1,T}\left(\kappa\right)\\
\Delta_{2,T}\left(\kappa\right)
\end{array}\right)+\sqrt{T}\left(\begin{array}{c}
\Delta_{1}\left(z_{1,T},\lambda_{T},\kappa,F\right)\\
\Delta_{2}\left(z_{2,T},\lambda_{T},\kappa,F\right)
\end{array}\right).
\]
Under $\mathbf{H_{0}}$, due to the previous, we have that for any
$i=1,2$, $\kappa,\kappa_{T}\in\mathbb{K}$, and $\kappa_{T}\rightarrow\kappa$,
$\lim_{T\rightarrow\infty}K_{i,T}\left(\kappa_{T}\right)$ is almost
surely equal to 
\[
\begin{cases}
\Delta_{i}^{\star}\left(\kappa\right), & \left(z_{i},\lambda,\kappa,\right)\in\text{Int}\Gamma_{i}\\
+\infty, & \left(z_{i},\lambda,\kappa,\right)\notin\Gamma_{i},\Delta_{i}\left(z_{i},\lambda,\kappa,F\right)>0\\
-\infty, & \left(z_{i},\lambda,\kappa,F\right)\notin\Gamma_{i},\Delta_{i}\left(z_{i},\lambda,\kappa,F\right)<0
\end{cases}.
\]
Furthermore, for any compact $\mathcal{K}_{i}$ that contains $\kappa\in\mathbb{K}$
such that $\left(z_{i,T},\lambda_{T},\kappa,\right)$ eventually belongs
to the boundary of $\Gamma_{i}$ we have that almost surely, 
\[
\lim\inf_{T\rightarrow\infty}\inf_{\kappa\in\mathcal{K}_{i}}K_{i,T}\left(\kappa\right)\geq\inf_{\kappa\in\mathcal{K}_{i}}\Delta_{i}^{\star}\left(\kappa\right)+\lim\inf_{T\rightarrow\infty}\inf_{\kappa\in\mathcal{K}_{i}}\sqrt{T}\Delta_{i}\left(z_{i,T},\lambda_{T},\kappa,F\right)\geq\inf_{\kappa\in\mathcal{K}_{i}}\Delta_{i}^{\star}\left(\kappa\right).
\]
Hence due to Proposition 3.2.(ii)-(iii) (ch. 5, p. 337) of Molchanov
(2006), $\left(\begin{array}{c}
K_{1,T}\left(\kappa\right)\\
K_{2,T}\left(\kappa\right)
\end{array}\right)$ almost surely converges w.r.t.\ to the product topology of epi-convergence
over $\mathbb{K}$, and continuously over $A_{i}\times\mathbb{L}$
to $K\left(\kappa\right)=\left(\begin{array}{c}
K_{1}\left(\kappa\right)\\
K_{2}\left(\kappa\right)
\end{array}\right)$, with $K_{i}\left(\kappa\right)=\begin{cases}
\Delta_{i}^{\star}\left(\kappa\right), & \left(z_{i},\lambda,\kappa\right)\in\Gamma_{i}\\
-\infty, & \left(z_{i},\lambda,\kappa\right)\notin\Gamma_{i}
\end{cases}$. Since $\mathbb{K}$ is compact, Theorem 3.4 (ch. 5, p. 338) of Molchanov
(2006) implies that almost surely, 
\[
\inf_{\kappa\in\mathbb{K}}K_{i,T}\left(\kappa\right)\rightarrow\begin{cases}
\inf_{\kappa:\left(z_{i},\lambda,\kappa\right)\in\Gamma_{i}}\Delta_{i}^{\star}\left(\kappa\right), & \exists\kappa:\left(z_{i},\lambda,\kappa\right)\in\Gamma_{i}\\
-\infty, & \nexists\kappa:\left(z_{i},\lambda,\kappa\right)\in\Gamma_{i}
\end{cases},
\]
jointly over $i=1,2$. When $\Gamma_{i}$ is not empty, by Theorem
7.11 of Rockafellar and Wetts (2009), and using the same notations
(to streamline the proof) for the random elements defined in the relevant
enhanced probability space, the sequence $\left(\inf_{\kappa}\Delta_{i}\left(z_{i},\lambda,\kappa,\sqrt{T}F_{T}\right)\right)_{T}$
is also equi-upper semi-continuous. Due to the proof of Lemma \ref{aac-1}
below and the form of $\mathbf{H_{0}}$, we have that the above sequence
is almost surely bounded, and thereby Theorem 3.4 (ch. 5, p. 338)
of Molchanov (2006) implies that almost surely, 
\[
\sup_{z_{i},\lambda}\inf_{\kappa}\Delta_{i}\left(z_{i},\lambda,\kappa,\sqrt{T}F_{T}\right)\rightarrow\sup_{z_{i},\lambda}\inf_{\kappa\in\Gamma_{i}}\Delta_{i}\left(z_{i},\lambda,\kappa,\mathcal{G}_{F}\right).
\]
When $\Gamma_{i}$ is empty the limit is trivially $-\infty$. Reverting
from the Skorokhod representations to the original sequences and employing
the continuous mapping theorem we get the result.
\end{proof}
\begin{proof}[Proof of Theorem \ref{main2}]
The first result follows by a direct application of Theorem 3.5.1.i
of Politis et al. (1999) from the results of Proposition \ref{EAD},
and the limiting quantile function being continuous for all $\alpha\in\left(0,1\right).$
The second result follows similarly, by also considering the results
of the auxiliary Lemma \ref{aac-1}. For the second result, if $\mathbf{H_{a}}$
is true, for some $\lambda^{\star}\in\mathbb{L}-\mathbb{K},$ and
any $\kappa\in\mathbb{K}$, there exists some $i,z^{\star}\in A_{i}$
such that $\Delta_{i}\left(z,\lambda,\kappa,F\right)>0$. Then, we
have that 
\[
\xi_{T}\geq\inf_{\kappa\in\mathbb{K}}\Delta_{i}\left(z^{\star},\lambda^{\star},\kappa,\sqrt{T}\left(F_{T}-F\right)\right)+\sqrt{T}\inf_{\kappa\in\mathbb{K}}\Delta_{i}\left(z^{\star},\lambda^{\star},\kappa,F\right),
\]
and from arguments analogous to the ones used in the proof of Proposition
\ref{EAD}, we have that the first term in the rhs of the last display
is asymptotically tight, while from the arguments used in the proof
of Proposition \ref{M-equiv}, the second term in the rhs of the last
display diverges to $+\infty$. The result follows from the properties
of $b_{T}$. 
\end{proof}
\begin{proof}[Proof of Theorem \ref{thm:mod}]
The result follows exactly as in the proofs of Proposition \ref{EAD}
and Theorem \ref{main2} by noting first that the relevant hypo-epi
convergence concepts in the aforementioned proposition also hold for
the relevant function restricted to $A_{i}^{\left(T\right)}$ from
the results there and the definition of the Painleve-Kuratowski set
convergence, and that $\sup_{\lambda}\inf_{\kappa}\Delta_{i}\left(z,\lambda,\kappa,\mathcal{G}_{F}\right)$
has the same $\sup$ w.r.t.\ $z$ with its restriction to any dense
subset of $A_{i}$ due to the compactness of $\mathbb{L}$ and $\mathbb{K}$
and Theorem 3.4 (ch. 5, p. 338) of Molchanov (1999). 
\end{proof}

\subsection*{Auxiliary Lemmata}

The following are auxiliary lemmata used for the derivation of the
proofs of Proposition \ref{EAD} and Theorem \ref{main2}. 
\begin{lem}
\label{lem:emp_pr}Under Assumption \ref{MSDmix} 
\[
\left(\begin{array}{c}
\Delta_{1}\left(z_{1},\lambda,\kappa,\sqrt{T}\left(F_{T}-F\right)\right)\\
\Delta_{2}\left(z_{2},\lambda,\kappa,\sqrt{T}\left(F_{T}-F\right)\right)
\end{array}\right)\rightsquigarrow\left(\begin{array}{c}
\Delta_{1}\left(z_{1},\lambda,\kappa,\mathcal{G}_{F}\right)\\
\Delta_{2}\left(z_{2},\lambda,\kappa,\mathcal{G}_{F}\right)
\end{array}\right)
\]
as random elements with values on the space of $\mathbb{R}^{2}$-valued
bounded functions on\linebreak{}
 $\mathbb{L\times K\times R}_{-}\times\mathbb{R}_{++}$ equiped with
the sup-norm. The limiting process has continuous sample paths. 
\end{lem}
\begin{proof}
Let $\theta:=\left(\lambda,\kappa,z_{1},z_{2}\right)\in\Theta:=\mathbb{L\times K\times R}_{-}\times\mathbb{R}_{++}$,
$\rho$ any non zero element of $\mathbb{R}^{2}$, and consider $\Delta\left(\theta,\cdotp\right):=\rho_{1}\Delta_{1}\left(z_{1},\lambda,\kappa,\cdotp\right)+\rho_{2}\Delta_{2}\left(z_{1},\lambda,\kappa,\cdotp\right)$.
Notice that Theorem 7.3 of Rio (2013), due to Assumption \ref{MSDmix},
implies that $\sqrt{T}\left(F_{T}-F\right)\rightsquigarrow\mathcal{G}_{F}$.
This implies that $\sqrt{T}\left(F_{T}-F\right)$ also weakly hypo-converges
to $\mathcal{G}_{F}$ (see for example Knight (1999)). Both are upper
semi-continuous (usc) $\mathbb{P}$ a.s. and the space of usc functions
with the topology of epiconvergence can be metrized as complete and
separable (see again Knight (1999)). Due to separability and the Skorokhod
Representation Theorem (see for example Theorem 1 in Cortissoz (2007))
there exists a suitable probability space and random elements with
values in the aforementioned function space such that $f_{T}^{*}\overset{d}{=}\sqrt{T}\left(F_{T}-F\right)$,
$f^{*}\overset{d}{=}\mathcal{G}_{F}$, and $f_{T}^{*}\rightarrow f^{*}$
a.s.. Let $J\equiv\overline{\text{span}}\left\{ f_{T}^{*},f^{*},T=1,2,\cdots\right\} $
equipped with the metrizable topology of weak convergence.\footnote{Here $\overline{\text{span}}$ denotes the closure w.r.t\. the particular
topology of the linear span.} Consider $\Delta\left(\cdotp,\cdotp\right)$ restricted to $J$ with
values in the linear space of stochastic processes, equipped with
the topology of convergence in distribution, with values in the space
of bounded real functions defined on $\Theta$ equipped with the sup-norm.
From Assumption \ref{MSDmix}, Remark \ref{thm:wd}, Corollary 4.1,
and Theorem 7.3 of Rio (2013), we also have that 
\[
\sup_{\theta\in\Theta}\sup_{T}\mathbb{E}\left[\left(\Delta\left(\theta,\sqrt{T}\left(F_{T}-F\right)\right)\right)^{2}\right]+\sup_{\theta\in\Theta}\mathbb{E}\left[\left(\Delta\left(\theta,\mathcal{G}_{F}\right)\right)^{2}\right]<+\infty.
\]
The latter inequality along with Theorem 6.5.2 in Narici and Beckenstein
(2010), the metrization of convergence in distribution by the bounded
Lipschitz metric (see for example p. 73, van der Vaart and Wellner
(1996)) which is bounded from above by $\sup_{\theta}\mathbb{E}\left[\left(x-y\right)^{2}\right]$,
for $x,y$ members of the aforementioned space of processes, imply
that $\Delta\left(\cdotp,\cdotp\right)$ as restricted above is continuous.
Hence the CMT implies that $\Delta\left(\theta,f_{T}^{*}\right)\rightsquigarrow\Delta\left(\theta,f^{*}\right)$
which means that $\Delta\left(\theta,\sqrt{T}\left(F_{T}-F\right)\right)\rightsquigarrow\Delta\left(\theta,\mathcal{G}_{F}\right)$.
This and the Cramer-Wold Theorem imply the needed result. The final
assertion follows from $\sup_{\theta\in\Theta}\mathbb{E}\left[\left(\Delta\left(\theta,\mathcal{G}_{F}\right)\right)^{2}\right]<+\infty$,
the discussion in Example 1.5.10 of van der Vaart and Wellner (1996),
and the continuity of $\mathbb{E}\left[\left(\Delta\left(\theta,\mathcal{G}_{F}\right)\right)^{2}\right]$
w.r.t.\ $\theta$. 
\end{proof}
\begin{lem}
\label{aac-1}If $\xi_{\infty}$ is non-constant, and under Assumptions
\ref{MSDmix} and \ref{simpl_span}, the distribution of $\xi_{\infty}$
has support $\left[0,+\infty\right)$, its cdf is absolutely continuous
on $\left(0,+\infty\right)$, and it may have a jump discontinuity
at zero, of size at most $ch_{\mathbb{L}}\left(\mathbb{K}\right)$. 
\end{lem}
\begin{proof}
The result stems from Corollary \ref{cor:used} as long as the requirements
of Assumption \ref{as1} are satisfied and an appropriately bounding
$\eta$ is found. For $\Lambda=\mathbb{L}\times\mathbb{K}\times\left\{ 1,2\right\} \times\mathbb{R}_{-}\times\mathbb{R}_{++}$
where $\left\{ 1,2\right\} $ is considered equipped with the discrete
metric, we have that $X_{\lambda}=1_{1}\left(i\right)\Delta_{1}\left(z_{1},\lambda,\kappa,\mathcal{G}_{F}\right)+1_{2}\left(i\right)\Delta_{2}\left(z_{2},\lambda,\kappa,\mathcal{G}_{F}\right)$,
for $\lambda=\left(\lambda,\kappa,i,z_{1},z_{2}\right)$, has continuous
sample paths from the final assertion of Lemma \ref{lem:emp_pr}.
Then notice that 
\[
\mathbb{E}\left[\sup_{\Lambda}\left(X_{\lambda}^{2}\right)\right]\leq\sum_{i=1,2}\mathbb{E}\left[\sup_{\lambda\in\mathbb{L}}\sup_{\kappa\in\mathbb{K}}\sup_{z\in A_{i}}\Delta_{i}^{2}\left(z,\lambda,\kappa,\mathcal{G}_{F}\right)\right].
\]
From the zero mean Gaussianity of the processes involved, Remark \ref{thm:wd},
the packing numbers of $\Lambda\times\mathbb{R}$ being bounded by
a polynomial w.r.t.\ the inverted radii, Proposition A.2.7 of Van Der
Vaart and Wellner (1996) implies the subexponentiality of the distributions
of the suprema above, and thereby the existence of their second moments.
Hence Hypothesis 1 of Assumption \ref{as1} holds. Using the discussion
in Nualart (2006), immediately after the proof of Proposition 2.1.11
(p. 109) we have that Hypothesis 2 of Assumption \ref{as1} also holds
due to Assumption \ref{MSDmix}. Due to zero mean Gaussianity and
excluding $\mathbb{P}$-negligible events $\Delta_{i}\left(z,\lambda,\kappa,\mathcal{G}_{F}\right)$
is zero only when $\kappa=\lambda$ and it is at most only then that
$\xi_{\infty}$ has degenerate variance. Thereby, $\mathcal{T}=\left\{ 0\right\} $
and we can try to obtain a lower bound for $\xi_{\infty}$. From the
integration by parts formula for the Lebesgue-Stieljes integral and
Assumption \ref{MSDmix}, we get 
\[
\xi_{T}\geq\max_{i}\sup_{\lambda\in\mathbb{L}}\inf_{\kappa\in\mathbb{K}}\Delta_{i}\left(0,\lambda,\kappa,\sqrt{T}F_{T}\right)
\]
\[
\geq\eta_{T}:=\frac{1}{2}\frac{1}{\sqrt{T}}\left(\sup_{\lambda\in\mathbb{L}}\lambda^{Tr}-\sup_{\kappa\in\mathbb{K}}\kappa^{Tr}\right)\sum_{i=1}^{T}\left(Y_{i}-\mathbb{E}\left(Y_{0}\right)\right)
\]
\[
\rightsquigarrow\eta_{\infty}:=\frac{1}{2}\sup_{\lambda\in\mathbb{\mathbb{L}}}\lambda^{Tr}Z-\frac{1}{2}\sup_{\kappa\in\mathbb{K}}\kappa^{Tr}Z,
\]
where $Z\sim N\left(0_{n\times1},\mathbb{V}\right)$. Hence, $\xi_{\infty}\geq\eta_{\infty}\geq0.$

The previous inequality implies the applicability of Corollary \ref{cor:used}
for $c=0$. We obtain the result by estimating an upper bound for
$\mathbb{P}\left(\eta_{\infty}=0\right)$. From Assumption \ref{M-span}
and the non-degeneracy of $\mathbb{V}$ the latter probability equals
exactly the probability that the maximum of the random vector $Z$
occurs at a coordinate that represents an extreme point of $\mathbb{S}$
to which corresponds a common effective extreme point for $\mathbb{L}$
and $\mathbb{K}$ (w.r.t.\ $\mathbb{L}$), say $\lambda$, evaluated
at which $\lambda^{Tr}Z$ is maximal. Using Theorem 2 in chapter 3
(p. 37) of Sidak et al. (1999) by (in their notation) letting $p$
be the density of the $n$-variate standard normal distribution and
$q$ the density of $N\left(0_{n\times1},\mathbb{V}\right)$, along
with Definition \ref{def:char}, we get: 
\[
\mathbb{P}\left(\eta_{\infty}=0\right)\leq ch_{\mathbb{L}}\left(\mathbb{K}\right).
\]
\end{proof}


\begin{thebibliography}{10}
\bibitem{AHP} Arvanitis, S., Hallam, M. S., Post, T. and N. Topaloglou.
2018. Stochastic spanning. Forthcoming in the Journal of Business
and Economic Statistics (http://dx.doi.org/10.1080/07350015.2017.1391099).

\bibitem{at}Arvanitis, S., and N. Topaloglou. 2017. Testing for prospect
and Markowitz stochastic dominance efficiency. Journal of Econometrics
198(2), 253-270.

\bibitem{bar} Banz, Rolf W. 1981. The relationship between return
and market value of common stocks. Journal of Financial Economics
9(1), 3-18.

\bibitem{Baucells} Baucells, M., and F.H. Heukamp. 2006. Stochastic
dominance and cumulative prospect theory. Management Science 52, 1409-1423.

\bibitem{black} Black, F., Jensen, M. and M. Scholes. 1972. The capital
asset pricing model: some empirical tests, in M.C. Jensen (ed.). Studies
in the Theory of Capital Markets, Praeger: New York, 79-124.

\bibitem{blum} M.E. Blume and I. Friend. 1973. A new look at the
Capital Asset Pricing Model. Journal of Finance 28(1), 19-34.

\bibitem{ca} Carhart, M. 1997. On persistence in Mutual Fund Performance.
Journal of Finance 52(1), 57-82.

\bibitem{cs} Cass, D., and J. E. Stiglitz. 1970. The structure of
investor preferences and asset returns, and separability in portfolio
allocation: A contribution to the pure theory of mutual funds. Journal
of Economic Theory, 2(2), 122-160.

\bibitem{cort}Cortissoz, J. 2007. On the Skorokhod representation
theorem. Proceedings of the American Mathematical Society 135(12),
3995-4007.

\bibitem{key-1} DeMiguel, V., L. Garlappi and R. Uppal, 2009, Optimal
versus naive diversification: How inefficient is the 1/n portfolio
strategy?. Review of Financial Studies 22, 1915-1953.

\bibitem{Edwards} Edwards, K.D. 1996. Prospect theory: A literature
review, International Review of Financial Analysis 5, 18-38

\bibitem{fama3} Fama, E. and K. French. 1992. The Cross-Section of
Expected Stock Returns. Journal of Finance 47(2), 427-465.

\bibitem{fama} Fama, E. and K. French. 1993. Common Risk Factors
in the Returns on Stocks and Bonds. Journal of Financial Economics
33, 3-56.

\bibitem{fama2} Fama, E. and K. French. 1997. Industry costs of equity.
Journal of Financial Economics 43, 153-193.

\bibitem{fama4} Fama, E. and K. French. 2015. A five-factor asset
pricing model. Journal of Financial Economics 116, 1-22.

\bibitem{famaM} Fama, E.F. and J.D. MacBeth. 1973. Risk, return and
equilibrium: empirical tests. The Journal of Political Economy 81,
607-636.

\bibitem{fz}Francq, C., and J. M. Zakoian. 2011. GARCH models: structure,
statistical inference and financial applications. John Wiley \& Sons.

\bibitem{friedman} Friedman, M., and L. J. Savage. 1948. The utility
analysis of choices involving risk. Journal of Political Economy 56,
279-304.

\bibitem{Gonzalo} Gonzalo, J. and J. Olmo. 2014. Conditional Stochastic
Ddominance Tests in Dynamic Settings. International Economic Review
55(3), 819-838.

\bibitem{GHK} Guggenberger, P., Hahn, J., \& Kim, K. (2008). Specification
testing under moment inequalities. Economics Letters, 99(2), 375-378.


\bibitem{key-4} Hadar, J. and W.R. Russell. 1969. Rules for ordering
uncertain prospects. American Economic Review 59, 2-34.

\bibitem{key-5} Hanoch, G., and H. Levy. 1969. The efficiency analysis
of choices involving risk. Review of Economic Studies 36, 335-346.

\bibitem{Horvath} Horvath, L. and Kokoszka, P. and R. Zitikis. 2006.
Testing for Stochastic Dominance Using the Weighted McFadden-type
Statistic. Journal of Econometrics 133, 191-205.

\bibitem{hub}Huberman, G. and S. Kandel. 1987. Mean-Variance Spanning.
Journal of Finance 42, 873-888.

\bibitem{knight}Knight, K. 1999. Epi-convergence in distribution
and stochastic equi-semicontinuity. Working Paper, Department of Statistics,
University of Toronto.

\bibitem{kroll=00003D00003D000026Levy}Kroll, Y., and H. Levy. 1980.
Stochastic Dominance Criteria: A Review and Some New Evidence, in
Research in Finance, Vol. II, Greenwich: JAI Press, pp. 263-277

\bibitem{kuosman} Kuosmanen, T. (2004). Efficient diversification
according to stochastic dominance criteria. Management Science 50(10),
1390-1406.

\bibitem{Levyp}Levy, H. 1992. Stochastic Dominance and Expected Utility:
Survey and Analysis. Management Science 38, 555-593.

\bibitem{Levyb}Levy, H. 2015. Stochastic dominance: Investment decision
making under uncertainty. Springer.

\bibitem{LL} M. Levy and H. Levy. 2002. Prospect Theory: Much Ado
about Nothing?. Management Science 48, 1334-1349.

\bibitem{LL2} Levy, H., M. Levy. 2004. Prospect theory and mean-variance
analysis. Review of Financial Studies 17(4), 1015-1041.

\bibitem{Lif}Lifshits, M. A. 1983. On the absolute continuity of
distributions of functionals of random processes. Theory of Probability
and Its Applications 27(3), 600-607.

\bibitem{linton1} Linton, O., Maasoumi, E. and Y.-J. Whang. 2005.
Consistent Testing for Stochastic Dominance under General Sampling
Schemes. Review of Economic Studies 72, 735-765.

\bibitem{lpw}Linton, O., Post, T., and Whang, Y. J. 2014. Testing
for the stochastic dominance efficiency of a given portfolio. The
Econometrics Journal 17(2), 59-74.

\bibitem{McFadden}McFadden, D. 1989. Testing for Stochastic Dominance,
in Studies in the Economics of Uncertainty, eds. T. Fomby and T. Seo,
New York: Springer- Verlag, pp. 113\textendash 134.

\bibitem{Molch}Molchanov, I. 2006. Theory of random sets. Springer
Science and Business Media.

\bibitem{Mosler=00003D00003D000026Scarcini}Mosler, K., and M. Scarsini.
1993. Stochastic Orders and Applications, a Classified Bibliography,
Berlin: Springer-Verlag.

\bibitem{nar}Narici, L., and E. Beckenstein. 2010. Topological vector
spaces. CRC Press.

\bibitem{nov}Novy-Marx, R. 2016. Understanding defensive equity.
Working Paper, Simon Graduate School of Business, University of Rochester and NBER.

\bibitem{nual}Nualart, D. 2006. The Malliavin calculus and related
topics. Berlin: Springer.

\bibitem{pol}Politis, D. N., J. P. Romano and M. Wolf. 1999. Subsampling.
Springer New York.

\bibitem{post}Post, T. 2003. Empirical Tests for Stochastic Dominance
Efficiency. The Journal of Finance 58: 1905-1931.

\bibitem{pk}Post, T., and M. Kopa. 2013. General linear formulations
of stochastic dominance criteria. European Journal of Operational
Research 230(2), 321-332.

\bibitem{pl}Post, T., and H. Levy. 2005. Does risk seeking drive
stock prices? A stochastic dominance analysis of aggregate investor
preferences and beliefs. Review of Financial Studies 18(3), 925-953.

\bibitem{Rein} M.R. Reinganum. 1981. A new empirical perspective
on the CAPM. Journal of Financial and Quantitative Analysis 16(4),
439-462.

\bibitem{Rio}Rio, E. 2013. Inequalities and limit theorems for weakly
dependent sequences. 3`eme cycle. pp.170. \textless{}cel-00867106\textgreater{}.

\bibitem{key-6} Rothschild, M. and J.E. Stiglitz. 1970. Increasing
Risk: I. A definition. Journal of Economic Theory 2(3), 225-243.

\bibitem{st}Scaillet, O., and N. Topaloglou. 2010. Testing for stochastic
dominance efficiency. Journal of Business and Economic Statistics
28(1), 169-180.

\bibitem{key-33}Simaan, Y., 1993, Portfolio selection and asset pricing-three-parameter
framework. Management Science 39, 568-577.

\bibitem{rank}Sidak, Z., P. K. Sen, and J. Hajek. 1999. Theory of
rank tests. Academic Press.

\bibitem{tob}Tobin, J. 1958. Liquidity Preference as Behavior Towards
Risk. Review of Economic Studies 25, 65-86.

\bibitem{van}van der Vaart, A. W., and J. A. Wellner. 1996. Weak
Convergence. Springer New York. 

\bibitem{key-3} Ziemba, W., 2005, The symmetric downside risk sharpe
ratio. Journal of Portfolio Management 32, 108-122. 
\end{thebibliography}
\end{document}